\documentclass[graybox]{svmult}

\usepackage{type1cm}        
\usepackage{makeidx}         
\usepackage{graphicx}        
\usepackage{multicol}        
\usepackage[bottom]{footmisc}

\usepackage{newtxtext}       %
\usepackage[varvw]{newtxmath}       

\makeindex             

\pdfoutput=1 

\usepackage[T1]{fontenc} 
\usepackage{psfrag}
\usepackage{mathrsfs}
\usepackage{latexsym}

\usepackage{mathtools}

\newcommand{\rf}[1]{(\ref{#1})}
\newcommand{\oh}{\frac{1}{2}}

\newcommand{\beq}{\begin{equation}}
\newcommand{\eeq}{\end{equation}}
\newcommand{\bea}{\begin{eqnarray}}
\newcommand{\eea}{\end{eqnarray}}
\newcommand{\beas}{\begin{eqnarray*}}
\newcommand{\eeas}{\end{eqnarray*}}
\newcommand{\beqs}{\begin{displaymath}}
\newcommand{\eeqs}{\end{displaymath}}

\newcommand{\br}{\langle}
\newcommand{\kt}{\rangle}

\newcommand{\cT}{{\cal T} }

\newcommand{\ben}{\begin{equation}}
\newcommand{\een}{\end{equation}}

\newcommand{\bdm}{\begin{displaymath}}
\newcommand{\edm}{\end{displaymath}}

\newcommand{\C}[1]{{\mathcal{#1}}}

\newcommand{\abs}[1]{\vert #1\vert}

\newcommand{\expect}[2]{\left\langle \, #1\, \right\rangle_{#2}}

\newcommand{\half}{\frac{1}{2}}
\newcommand{\third}{\frac{1}{3}}

\newcommand{\quarter}{\frac{1}{4}}

\newcommand{\lderiv}[1]{#1\frac{\partial }{\partial #1}}
\newcommand{\asqrt}[1]{{#1}^\half}
\newcommand{\absd}[1]{\|{#1}\|}

\renewcommand{\twoheadrightarrow}{\mathrel{\mathrlap{\rightarrow}
\mkern1mu\rightarrow}}

\renewcommand{\bar}{\overline}
\renewcommand{\tilde}{\widetilde}
\renewcommand{\hat}{\widehat}

\newcommand{\Cfinite}{{\mathcal C}_{\text{fin}}}

\begin{document}

\title*{From Trees to Gravity }
\author{Bergfinnur Durhuus,  Thordur Jonsson and John Wheater}
\institute{Bergfinnur Durhuus \at Department of Mathematical Sciences, University of Copenhagen, Universitetsparken 5, 2100 Copenhagen, Denmark,  \email{durhuus@math.ku.dk}
\and Thordur Jonsson \at The Science Institute, 
University of Iceland, Dunhaga 3, 107 Reykjavik, Iceland, \email{thjons@hi.is}
\and John Wheater \at Rudolf Peierls Centre for Theoretical Physics, Department of Physics, University of Oxford, Parks Road, Oxford OX1 3PU, UK, \email{john.wheater@physics.ox.ac.uk} }
%
%
\maketitle

\abstract{In this article we study two related models of quantum geometry: generic random trees
and two-dimensional causal triangulations.  The Hausdorff and spectral dimensions that arise
in these models are calculated and their relationship with the structure of the underlying
random geometry is explored.  Modifications due to interactions with matter fields are also
briefly discussed.  The approach to the subject is that of classical statistical mechanics and
most of the tools come from probability and graph theory.}

\keywords{Causal triangulation, Graph theoretic, Tree correspondence, Generic trees, Spectral dimension, Hausdorff dimension, Scaling amplitude}

\medskip

\noindent{\bf Note} {This is a contribution to the \emph{Handbook of Quantum Gravity} which will be published in the beginning of 2023. It will appear as a chapter in the section of the handbook entitled \emph{Causal Dynamical triangulations}.}

\section{Introduction}\label{intro}

In this contribution we adopt a graph theoretic and probabilistic perspective on two-dimensional Causal Dynamical Triangulations ({CDT}s). We consider causal triangulations (CTs) as a particular class of planar graphs (defined in Section \ref{subsec:CDT}) that are distributed according to a probabilistic law. Our primary goal is then to apply tools from graph and probability theory to analyse the large scale geometric properties of these models. In particular we exploit, in a variety of contexts, the correspondence (established in Section \ref{subsec:bijection}) between  CTs and planar tree graphs (defined in Section \ref{sec:tree}). In the process we will define and study the generic random tree model, both because of its relation to branching processes and CTs and because of its independent interest as a testing ground for investigating various aspects of random graph models in general.

We will consider two different ensembles of CTs. The first, much studied in the literature, is the grand canonical ensemble (defined in Section \ref{sec:gce}) which is based on an expansion in the size of finite graphs. We use the planar tree correspondence to give an alternative account of the scaling behaviour of loop correlation functions in the vicinity of the radius of convergence and determine the associated scaling Hausdorff dimension $d_H$.
 The second ensemble, to date less studied, is based on infinite CTs and can be thought of as an infinite volume limit suitable for studying local geometric characteristics of typical CTs. Again, we use the planar tree correspondence to investigate, in Sections \ref{sec:hausdorff} and \ref{sec:spectraldim}, the fractal properties of CTs; in particular, we determine the local Hausdorff dimension $d_h$ as well as the spectral dimension $d_s$ (defined in Section \ref{sec:5.1}).  Indeed, for CTs all three dimension exponents, $d_H, d_h$ and $d_s$ have the value $2$. That this is not generally the case is illustrated by the generic trees for which $d_h=d_H=2$, but whose spectral dimension is $d_s=\frac{4}{3}$ as shown in Section \ref{sec:5.3}. 

It is obviously of interest to understand the extent to which the results outlined above are universal, and to investigate how they extend to the more general case of CTs coupled to statistical mechanical systems such as dimers and Ising spins. A few remarks on the rather sparse existing analytical results in this direction are collected in Section \ref{sec:matterfields}. Finally in Section \ref{sec:wherenext} we draw together the possible avenues for future research. 

\section{Preliminary on probability\label{sec:prelim}}

Before embarking on the main subject it is worth pointing out the difference of viewpoints represented by the two ensembles discussed above. In the grand canonical ensemble, the quantities of interest are defined as sums over graphs of finite size which are each attributed a positive finite weight. The scaling limit then involves adjusting the coupling constants so that large surfaces yield the dominant contribution to the statistical sums involved.  This procedure enables the construction of the continuum limit of certain correlation functions, but does not construct the limiting distribution of continuum surfaces -- although that would be an important achievement, see \cite{Bjornberg:2022} and references therein.
On the other hand, the infinite volume limit does involve the construction of a probability distribution of infinite CTs as a limit of distributions of finite CTs of fixed volume. 

A simple illustration of the basic philosophy of this construction is obtained by considering standard random walk on the hypercubic integer lattice $\mathbb Z^d$. Here, a walk (or path) $\omega$  is a sequence (finite or infinite) of points $\omega_0,\omega_1,\omega_3,\dots$ in $\mathbb Z^{d}$ such that $\omega_i$ and $\omega_{i+1}$ have distance $1$ for all $i$. If $\omega$ consists of $N+1$ points we say it has length $N$ and denote it by $|\omega|$. One then attributes a weight $p(\omega) = (2d)^{-|\omega|}$ to a finite path $\omega$ starting at, say, the origin $0$. Considering only paths $\omega$ of a fixed length $N$, the weights sum up to $1$, i.e. $p$ defines a probability distribution $p_N$ on paths of length $N$. Moreover, these distributions are  clearly consistent for different values of $N$ in the sense that if $N<M$ and we consider the set consisting of all paths $\omega^\prime$ of length $M$ coinciding with a given path $\omega$ in the first $N$ steps, then the weights of those paths $\omega^\prime$  add up to $p(\omega)$. This property leads in a natural way to a probability distribution on the space of infinite paths starting at $0$ as follows. Given an infinite path $\omega$, let ${\cal B}_{\frac 1N}(\omega)$ denote the set of finite or infinite paths coinciding with $\omega$ in the first $N$ steps and define the probability of this set to be 
$$
P({\cal B}_{\frac 1N}(\omega))= (2d)^{-N}\,.
$$
Recalling that $p_M$ attributes a weight $0$ to all paths not of length $M$, the compatibility property implies that $P({\cal B}_{\frac 1N}(\omega))$
%
is determined by the finite size distributions as 
\begin{equation}\label{limmeas}
P({\cal B}_{\frac 1N}(\omega)) = \lim_{M\to\infty} p_M({\cal B}_{\frac 1N}(\omega))\,.
\end{equation}
It is convenient to regard ${\cal B}_{\frac 1N}(\omega) $ as a ball of radius $\frac 1N$ around $\omega$ in the space $\Omega$ of all paths (finite or infinite) starting at $0$, where the distance between any two different  paths $\omega, \omega^\prime$  is defined as
$$
d(\omega,\omega^\prime) = \frac{1}{N+1}\,,
$$
with $N=\mbox{\rm max}\{n\mid \omega_0=\omega^\prime_0, \dots,\omega_n=\omega^\prime_n$\}. It is a consequence of a rather general result on sequences of probability measures, details of which can be found in \cite{billingsley}, that the limiting values of ball probabilities as given by \eqref{limmeas} uniquely specify  a probability distribution on $\Omega$, thus defining an ensemble of random paths, called random walk in $\mathbb Z^d$. By letting $N$ grow large in the preceding discussion, it should be clear that individual paths have vanishing $P$-probability. Since the set of finite paths is countable, it follows that the whole set of finite paths has vanishing $P$-probability, which is also expressed by saying that random walks are almost surely infinite, or that $P$ is concentrated on infinite walks. A stronger statement is that random walks are almost surely unbounded (in $\mathbb Z^d$), which we leave for the reader to figure out (it follows from the fact that random walk in a finite connected graph is recurrent, in the language of Section \ref{sec:5.1}). 

The strategy for constructing the infinite volume limit of generic trees in Section \ref{sec:tree}, and of CTs in Section \ref{subsec:bijection}, follows the procedure sketched above quite closely. It is a recurrent theme in Sections \ref{sec:hausdorff} and \ref{sec:spectraldim} to establish certain properties possessed by \emph{typical} CTs, since the existence of atypical ones occurring with vanishing probability must be expected. A standard tool for establishing such 
properties is the Borel-Cantelli lemma. To formulate this result, recall first that an \emph{event} in a probability space $\Omega$ is simply a subset $A$ of $\Omega$ with an associated probability $P(A)\geq 0$. In particular, $P(\Omega)=1$ and we say that an event $A$ occurs \emph{almost surely} (a.s.) if $P(A)=1$, which is equivalent to $P(\Omega\setminus A)=0$. Moreover, probabilities of mutually exclusive events add up to the probability of their union, i.e. if $A_1, A_2,\dots$ are events such that $A_i\cap A_j =\emptyset$ for $i\neq j$ then  
 $$
    P\left(\bigcup_i A_i\right)=\sum_i P(A_i)\,.
    $$
In the example of random walk above, the set 
$$
E_R = \{\omega :
\|\omega_n\|\leq R\; \mbox{for all $n$}\}\,,
$$
where $\|\cdot\|$ denotes the Euclidean distance to the origin in $\mathbb Z^{d}$, is the event that a walk is confined to the ball of radius $R$ around the origin and $\cup_{R=1}^\infty E_R$ is the event consisting of bounded walks. 

Now let $A_1, A_2, \ldots$ be an arbitrary sequence of events
and suppose 
$$
\sum_{n=1}^\infty  P(A_n)<\infty .
$$
Then, since 
$$
P(A_k\cup A_{k+2}\cup A_{k+3}\dots)\, \leq\,\sum_{n=k}^\infty P(A_n)\,,
$$
this implies that
$P(A_k\cup A_{k+2}\cup A_{k+3}\dots)\,\to\, 0$  as $k\to\infty$,
i.e., the probability that $A_n$ occurs for some $n\geq k$ tends to zero as $k$ grows large. Hence,
the probability that infinitely many of the events $A_n$ happen
is 0, which is the content of the Borel-Cantelli lemma.  

In our applications we typically use the Borel-Cantelli lemma to bound the size of certain graphs.  For example, suppose $\Omega$ consists of infinite graphs and for any given $n$ we associate a finite graph $G_n$ to each $G\in\Omega$ and let $|G_n|$ denote the number of edges in $G_n$. If 
$(a_n)$ is a sequence of positive numbers and
$$
\sum_{n=1}^\infty P(\{G: 
|G_n|>a_n\})\,<\,\infty\,,
$$
the Borel-Cantelli lemma implies that a.s. $|G_n|>a_n$ only 
for a finite number of $n$'s
and hence, $|G_n|\leq a_n$ for $n$ large enough with probability $1$.

\section{Random trees}\label{sec:tree}

Many of the discrete structures used to model quantum gravity can be viewed 
as graphs or graphs with some added structure.  We recall that a 
graph $G$ is a 
set of vertices $V(G)$ and a set of edges $E(G)$ which are unordered pairs of 
distinct vertices. 
In physics applications the edges are sometimes called links.
The graph is finite if the number of vertices is finite, otherwise it is infinite.
The {\it degree} (sometimes called order or valency) of 
a vertex $v$ is the number of edges containing $v$, which will be denoted by $\sigma_v$. We say that a graph is rooted if one vertex is singled out and called the {\it root}.

A {\it path} of length $n\geq 1$ in a graph is a sequence of oriented edges 
$$
(v_0,v_1), (v_1,v_2), (v_2,v_3) \ldots (v_{n-1},v_n).
$$
If $\eta$ is a path we use the notation $v\in \eta$ to indicate that one of the edges in $\eta$ has $v$ as an endpoint.
   We say that the path is {\it closed} if
$v_0=v_n$ and {\it simple} if 
all the vertices $v_i$ are different except possibly $v_0$ and $v_n$ which happens when the 
path is closed. A simple and closed path will be called a \emph{cycle} and two cycles are considered identical if one is obtained from the other by a cyclic permutation of the edges. A trivial path consists by definition of a single vertex and is considered closed and simple and of length $0$.

A graph is connected if there is a path between any two vertices.
The distance (also called graph distance) between two vertices in a connected 
graph is the length of the shortest path  
connecting them. 
The graph spanned by a subset $V_0$ of vertices of a given graph $G$ consists of $V_0$ itself and those edges in $E(G)$ that connect the
vertices of $V_0$. 
The ball of radius $R$ centred at a vertex $v$ of $G$ is the sub-graph of $G$ spanned by the vertices at graph distance less than or equal to $R$ from $v$ and will be denoted by $B_R(G;v)$. If $v$ is the root of $G$ the reference to $v$ will in general be omitted. Similarly, the boundary of $B_R(G;v)$, i.e. the sub-graph spanned by the vertices at distance $R$ from $v$ will be called $S_R(G;v)$.
The size $|G|$ of a graph $G$ is the number of elements in $E(G)$.
The number of vertices in $G$ will be denoted $\|G\|$ and, if $v$ is the root of $G$, we set $D_R(G)=\|S_R(G;v)\|$.  

We define a tree to be a connected
graph which does not contain any cycle.  We say that a graph is planar
if it is embedded in the plane such that no two edges intersect.   Note that most graphs cannot
be embedded in this way
but all trees can and a given tree can in general be embedded in many different ways.
We use the term planar tree to refer to a tree with a fixed embedding up to orientation preserving 
homeomorphisms of the plane. Alternatively, a planar tree can be defined as a purely combinatorial 
object, see e.g.\ \cite{chassaing:2006}.
 
 \subsection{The generic random tree}\label{sec:gentree}
Let $\cT$ be the set of all
planar rooted trees, finite or infinite, such that the root, $r$, is
of degree 1 and all vertices have finite degree. 
 The subset
of $\cT$ consisting of trees of fixed size $N$ will be denoted
$\cT_N$.  The subset consisting of the infinite trees will be denoted 
$\cT_\infty$.  Given a tree $T\in\cT$ and non-negative integer $R$, the
ball $B_R(T)$ is again a rooted planar tree and its boundary $S_R(T)$ consists of $D_R(T)$ isolated vertices, whose distance $R$ from the root will also be called their height in $T$. We let 
$T\setminus r$ denote the set of all vertices in $T$ except the root and note that $\|T\|=|T|+1$ for any finite tree $T$.

For later use we define the distance
$d_{\cT}(T,T')$ between two 
trees $T, T'$ as $R^{-1}$, where $R$ is the
radius of the largest ball around the root  common to $T$ and
$T'$. We can view $\cT$ as a metric space with metric
$d_{\cT}$, see \cite{Durhuus20034795} for some of its properties. In particular, for any positive integer $R$ the ball in $\cT$ 
of radius $R^{-1}$ is given by 
\beq\label{defBall}
{\cal B}_{\frac 1R}(T) = \{T^\prime : 
B_R(T^\prime) = B_R(T)\}\,,
\eeq
i.e., it consists of all trees coinciding with $T$ up to height $R$. If $T$ has $K=D_R(T)$ vertices at 
height $R$, the trees in ${\cal B}_{\frac 1R}(T)$ are obtained by grafting arbitrary trees $T_1,\dots, T_K$ onto those $K$ vertices, i.e., identifying the  root edge of $T_i$ with the edge in $T$ incident on the $i$'th vertex in $S_R(T)$. In this way, there is a one-to-one correspondence 
between trees in ${\cal B}_{\frac 1R}(T)$ and $K$-tuples of trees in ${\cal T}$, that we shall denote by 
\beq\label{ballcorresp}
F_R: {\cal B}_{\frac 1R}(T)\,\to\, {\cal T}^{K}\,.
\eeq

An ensemble of random graphs is a set of graphs equipped with a probability
measure.
In this section we define a class of ensembles of infinite trees of relevance to quantum gravity, 
called \emph{generic tree ensembles}, and discuss some of their properties.  We proceed 
by first defining probability measures
on finite trees of fixed size $N$ and then 
showing how to obtain a limit as $N$ tends to infinity. 

Given a set of non-negative {\it branching weights} $w_n,\,n\in\mathbb
N$, we define, in the spirit of classical statistical mechanics, 
the {\it finite volume partition functions}, $Z_N$, by
\beq\label{x1}
Z_N = \sum_{T\in\cT_N}\prod_{v\in T\setminus r}w_{\sigma_v}\,.
\eeq
 We assume $w_1>0$, since $Z_N$ vanishes
otherwise, and we also assume $w_n>0$ for some $n\geq 3$ since
otherwise only the linear chain of length $N$ would contribute to $Z_N$.
Under these assumptions the generating function $\phi$ for the branching weights,
\beq\label{geng}
\phi (z) = \sum_{n=1}^\infty w_n z^{n-1}\,,
\eeq
is strictly increasing and strictly convex on the interval $[0,z_c )$ with $\phi(0)=w_0$,
where $z_c$ is the radius of convergence for the series (\ref{geng}),
which we assume is positive.

It is well known, see e.g. \cite{bookADJ}, that the generating function
for the finite volume partition functions,
\beq\label{genZ}
Z(\zeta ) = \sum_{N=1}^\infty Z_N \zeta^N\,,
\eeq
satisfies the equation
\beq\label{fixZ}
Z(\zeta ) = \zeta \phi (Z(\zeta ))\,,
\eeq
where $\zeta$ is the fugacity associated to each edge. The proof of this 
identity is illustrated in Fig.\ref{fixequation}. If  the degree of the vertex next to the root is $n$, then
there are $n-1$ trees attached to it in addition to the root edge, which has weight $\zeta$.  
Summing over $n$ then yields equation \rf{fixZ}.

\begin{figure}[ht!]
\sidecaption[t]
\includegraphics[scale=1.0]{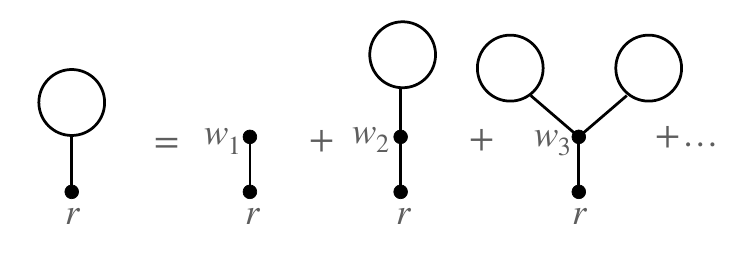}
\caption{A graphical illustration of the derivation of equation (\ref{fixZ}).}
\label{fixequation}
\end{figure}

From the properties of $\phi$ it follows that the straight line in the 
plane through the origin with slope $\zeta^{-1}$ intersects the graph of $\phi$ at least once 
(and at most twice) for $\zeta>0$ small enough, see Fig.\ref{Phi}. By \eqref{fixZ} and the fact that $Z(0)=0$ it 
follows that $Z(\zeta)$ is determined by the first intersection point for $\zeta$ small 
enough and that the solution persists for $\zeta<\zeta_0$ where $\zeta_0$ is the radius of 
convergence of the series (\ref{genZ}). 
Since $Z$ is an increasing function of $\zeta$, 
the limit
\beq\label{x2}
Z_0 = \lim_{\zeta\uparrow \zeta_0}Z(\zeta )
\eeq
is finite and $\leq z_c$. In the following we consider the \emph{generic} case where
\beq\label{genass}
Z_0<z_c\,.
\eeq
In this case, 
it follows from \eqref{fixZ} that the slope of the line through the origin that is tangent to the 
graph of $\phi$ equals $\zeta_0^{-1}$, i.e.
 $Z_0=Z(\zeta_0)$ is the unique positive solution to the equation
\beq\label{eqcrit1}
Z_0\phi'(Z_0) = \phi(Z_0)\,.
\eeq
The inequality (\ref{genass}) is the condition on the branching weights which singles out the generic ensembles of 
infinite trees to be defined below.  In particular, all sets of branching weights with infinite 
$z_c$ define a generic ensemble.

\begin{figure}[ht!]
\sidecaption[t]
\includegraphics[scale=1.0]{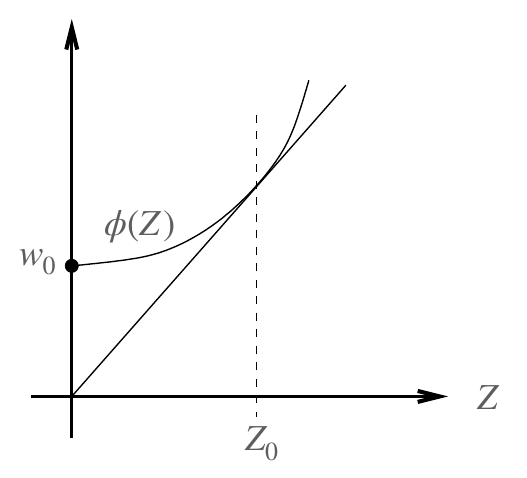}
\caption{An illustration of the intersection between the graph of $\phi$ and a straight line through the origin in the generic case.}
\label{Phi}
\end{figure}

In the special case $w_n=1$ for all $n$, that will be encountered frequently in this article, we have by \eqref{x1} that $Z_N=\sharp\cT_N$\footnote{We use $\sharp A$ to denote the number of elements in a set $A$.} and evidently $\phi(z)=(1-z)^{-1}$ by \eqref{geng}. Solving \eqref{fixZ} then gives 
\beq\label{ZZZ}
Z(\zeta )=\oh -\oh\sqrt{1-4\zeta },
\eeq
so $\zeta_0=\quarter$
and $Z_0=\oh$.  Taylor expanding this expression now yields
\beq\label{Catalan}
Z_N= C_{N-1} := \frac{(2N-2)!}{ N! (N-1)!}\,,
\eeq
where $C_N$ is the $N$th Catalan number.

Under the assumptions on the branching weights and assuming
(\ref{genass}) we may, in the general case,  Taylor expand $\phi$ around $Z_0$ in
(\ref{fixZ}) and use (\ref{eqcrit1}) to conclude that the analytic function $Z(\zeta)$ has a square root branch point at $\zeta_0$ and is given by
\beq\label{singZ}
Z(\zeta ) = Z_0 - \sqrt{\frac{2\phi(Z_0)}{\zeta _0\phi''(Z_0)}}\sqrt{\zeta_0-\zeta}+
O(\zeta _0-\zeta)\,,
\eeq
where the square root is chosen to be positive for $\zeta<\zeta_0$.
The asymptotic behaviour of $Z_N$ is then given by
\beq\label{eqn:flajolet}
Z_N= \sqrt{\frac{\phi(Z_0)}{2\pi \phi''(Z_0)}}N^{-\frac{3}{2}}\zeta_0^{-N}(1+O(N^{-1})),
\eeq
provided $Z_N\neq 0$.
The proof of this 
result can be found in \cite{flajolet:2009}, Sections VI.5
and VII.2.  In the special case of $w_n=1$ for all $n$ (\ref{eqn:flajolet})
follows easily from (\ref{Catalan}) using Stirling's formula.

We define the probability distribution $\nu_N$ on $\cT _N$ by
\beq\label{p1}
\nu_N(T) = Z_N^{-1}\prod_{v\in T\setminus
r}w_{\sigma_v}\,,\quad T\in\cT_N\,,
\eeq
assuming $Z_N\neq 0$.
Using the correspondence $F_R: {\cal B}_{\frac 1R}(T)\,\to\,{\cal T}^{K}$ described above, where $K=D_R(T)$, the probability $\nu_N({\cal B}_{\frac 1R}(T))$ can be expressed in terms of partition functions, and so the preceding results about the asymptotic behaviour of $Z_N$ can be applied to determine the limiting probability as $N\to\infty$.
More precisely, one obtains (see Appendix A in \cite{Durhuus:2007} for details) 
\beq\label{ballprob}
\lim_{N\to\infty} \nu_N({\cal B}_{\frac 1R}(T)) = D_R(T)\,Z_0^{D_R(T)-1}\zeta_0^{|B_{R-1}(T)|} \prod_{v\in B_{R-1}(T)\setminus r}w_{\sigma_v}.
\eeq
Similarly to the case of random
walk in $\mathbb Z^d$ discussed in Section \ref{sec:prelim}, this equals the volume of ${\cal B}_{\frac 1R}(T)$ with respect to a limit probability measure concentrated on ${\cal T}_\infty$, which we denote by $\nu$.
We call the
ensembles $(\cT_\infty ,\nu)$ defined in this way {\it generic ensembles}, referring back to the genericity assumption (\ref{genass}).
The expectation with respect to $\nu$ will be denoted $\langle\cdot\rangle _{\nu}$.

Note that if $w_n=1$ for all $n$, then $\nu_N$ is the uniform measure on trees of size $N$, i.e. 
\begin{equation}\label{Nuniftree}
 \nu_N(T) = \frac{1}{C_{N-1}},\quad T\in \cT_N\,.
\end{equation}
In this case, $\nu$ is called the Uniform Infinite Planar Tree (UIPT). 

The expression \eqref{ballprob} has a simple formal interpretation which we now explain. Given an infinite tree $T$ a {\it spine} for $T$ is an infinite linear chain (non-backtracking path) in $T$ starting at the root. We claim that $\nu$ is concentrated on the subset $\cal S$ of $\cal T$ consisting of trees with a single spine. Thus, if we denote the vertices on the spine
by $s_0=r, s_1, s_2,\dots$, ordered by increasing distance from the root, the trees in ${\cal S}$ are obtained by attaching branches, i.e. finite trees in $\cT$, to each spine vertex $s_n$ except the root, by identifying their roots
with $s_n$. If $s_n$ is of degree $\sigma$ there are $\sigma-2$
branches attached to the spine at $s_n$.
For $T\in{\cal S}$ we associate a weight $w_{\sigma_v}$ to each vertex of $T$ different from the root and a weight $\zeta_0$ to each edge of $T$. We now argue that on a formal level these assignments characterise $\nu$. Considering the $D_R(T)=K$ vertices of $T$ at height $R\geq 1$, it is clear that one of them, say the $i$'th from the left, must be $s_R$ and for the corresponding $K$-tuple $F_R(T)=(T_1,\dots,T_K)$ this means that  $T_i$ belongs to ${\cal S}$ while $T_j$ is a finite tree in ${\cal T}$ for $j\neq i$.
Using the weight assignments specified above, we obtain the total weight of single spine trees in ${\cal B}_R(T)$ by summing over all possible $K$-tuples $(T_1,\dots,T_K)$. This yields a factor $Z_0$ for each $j\neq i$, while the sum over $T_i$ is to be interpreted as an integral over ${\cal S}$ with respect to $\nu$ which gives $1$. Moreover, summing over the position of $i$ yields the factor $D_R(T)$ while the remaining factors in \eqref{ballprob} arise from the part of $T$ below height $R$, thus providing the desired interpretation of \eqref{ballprob}.  

\begin{figure}[ht!]
\sidecaption[t]
\includegraphics{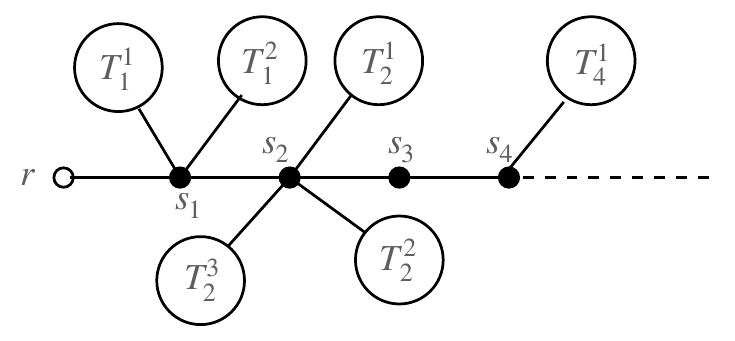}
\caption{The first few vertices $s_i$ on the spine of a tree and the finite branches $T_i^n$
attached to them.}
\label{fig1}
\end{figure}

Using the above description of single spine trees in terms of (finite) branches attached 
at spine vertices (see Fig. \ref{fig1}), we can obtain the probability $P$ 
that the spine vertices $s_1,\dots, s_n$ have given degrees $\sigma_1,\dots,\sigma_n$, respectively, by summing the attributed weights over the branches attached to these vertices as well as the infinite tree spanned by $s_n$ and $s_{n+1}$ and its descendants. Since the $\sigma_i-2$ branches attached to $s_i$ can be divided in $\sigma_i-1$ ways into left and right branches and summation over individual branches yields a factor $Z_0$, we obtain  
$$
P=\prod_{i=1}^n \zeta_0(\sigma_i-1) w_{\sigma_i}Z_0^{\sigma_i-2}\,.
$$
Noting that 
\beq\label{x6}
\sum_{k=2}^\infty \zeta_0(k-1)w_{k}Z_0^{k-2}\; =\; 
\zeta_0\phi'(Z_0)\;=1\,,
\eeq
where the last equality follows from (\ref{fixZ}) and \eqref{eqcrit1}, 
this shows that the degrees of spine vertices are independently and identically distributed with probability 
\beq\label{branchprob}
\varphi(k) = \zeta_0(k-1) w_k\,Z_0^{k-2}
\eeq
for having degree $k$.
Similarly, 
it follows from the interpretation of $\nu$ just given, that the individual branches $T$ are identically and independently distributed with probability proportional to $\zeta_0^{|T|}\prod_{v\in T\setminus r}w_{\sigma_v}$. The appropriate normalisation factor is $Z_0^{-1}$ yielding the probability distribution 
\beq\label{mucrit}
\mu(T)= Z_0^{-1}\zeta_0^{|T|}\prod_{v\in T\setminus r}w_{\sigma_v}
\eeq
for $T\in\cT$ finite.

Using \eqref{ballprob} and \eqref{branchprob} we can also determine the distribution of the total size of branches at a given spine vertex $s_i$, which will be needed in Section \ref{sec:spectraldim}. Thus, denoting the union of the branches $T_i^1,\dots,T_i^k$ at $s_i$ by $Br_i$, which clearly is a tree, we define  
$$
\bar Z_N = \nu(\{|Br_i| = N\}) = \sum_{k=0}^\infty \sum_{N_1+\dots+N_k=N}\varphi(k)\mu(\{|T_i^1|=N_1\})\cdots\mu(\{|T_i^k|=N_k\}),
$$
where the last equality follows from the independence of the distributions of $T_i^1,\dots,T_i^k$ and $\varphi(k)$ is given by \eqref{branchprob}. Using also \eqref{mucrit},
 the corresponding  generating function is given by
$$
\bar Z(s) = \sum_{N=0}^\infty \bar Z_N s^N = \sum_{k=0}^\infty \zeta_0(k+1)w_{k+2}Z(s\zeta_0)^k = \zeta_0\phi^\prime(Z(s\zeta_0))\,.
$$
As before, we may use the analyticity properties of $\phi$ to Taylor expand the last expression around $s=1$ and obtain 
\begin{equation}\label{bar-sing}
\bar Z(s) = 1- \bar c_0\sqrt{1-s} +O(|1-s|)
\end{equation}
for $s$ in a neighborhood of the unit circle, where  $\bar c_0>0$ is a constant.
Applying transfer theorems as before, see \cite{flajolet:2009}, 
this implies the asymptotic behaviour 
$$
\bar Z_N = \mbox{const.}\, N^{-3/2}\left(1+O(\frac 1N)\right)\,.
$$ 
for $N$ large, which will be used in Section \ref{sec:spectraldim}. 

\subsection{Trees and branching processes}
We will now show how
the probabilities  $\mu (T)$ defined in 
(\ref{mucrit})
arise from a {\it branching process}.
A Galton-Watson (GW) process is 
specified by a sequence $p_n,\,n= 0,1,2,\dots$, of non-negative numbers which
are called {\it offspring
probabilities} and satisfy 
\beq\label{prob1}
\sum_{n=0}^\infty p_n = 1\,.
\eeq
The number $p_n$ can be viewed as the probability of having $n$ offspring.  The process
begins with one individual who has $n$ offspring with probability $p_n$.   Each of the offspring
has $n$ descendants with the same probability distribution and the process continues in the same way
generation after generation.
Clearly it can stop after a finite number of steps or go on for ever.    The motivation of 
Galton and Watson was to find out how likely it was that families would die out.   In order to have a 
one to one correspondence between
trees generated by a GW process and the tree ensembles we have been discussing we have to 
assume that the first generation in the process has only one member since the root vertex has degree 1.

We say that the process is {\it critical} if the mean number of offspring is 
1, i.e.,         
\beq\label{GWcrit}         
\sum_{n=1}^\infty np_n = 1\,.         
\eeq 
A critical GW process gives rise to a 
probability distribution $\pi$ on the subset
of finite trees $T$ in ${\cal T}$ given by 
\beq\label{GWprob}
\pi (T) = \prod_{i\in T\setminus r}p_{\sigma_i-1}
\eeq
as a consequence of \eqref{x4} below.
If we take
\beq\label{GWprobw}
p_n = \zeta _0 w_{n+1}Z_0^{n-1},
\eeq  
where $w_n$, $\zeta_0$ and $Z_0$ correspond to a generic tree as described above, then
\beq\label{x3}
\sum_{n=0}^\infty p_n = \zeta _0 \sum_{n=0}^\infty w_{n+1}Z_0^{n-1}\; =\; 
\zeta _0Z_0^{-1}\phi(Z_0)\;=1\,,
\eeq
where the last equality follows from (\ref{fixZ}).  Furthermore, by (\ref{GWprob}) we have 
\beq\label{x4}
\pi (T)\; = \;\zeta _0^{|T|}\prod_{i\in T\setminus r}
w_{\sigma_i}Z_0^{\sigma_i-2}
\;= \;Z_0^{-1} \zeta _0^{|T|}\prod_{i\in T\setminus r}w_{\sigma_i} \;=\;\mu (T),
\eeq
since
\beq\label{x5}
\sum_{i\in T\setminus r}(\sigma_i-2)\;=\; -1
\eeq
for a tree $T$ with a root of degree $1$.  The reader may also easily verify 
that \eqref{x6} is equivalent to \eqref{GWcrit},
so the GW process defined by (\ref{GWprobw}) is critical.
Note that for the uniform tree we have 
\beq\label{uniform}
p_n=2^{-n-1}.
\eeq

In the following we let $f$ denote the generating function for the
offspring probabilities given by (\ref{GWprobw}),
\beq\label{genf}
f(s)= \sum_{n=0}^\infty p_ns^n \; =\; \zeta_0\sum_{n=1}^\infty
w_nZ_0^{n-2}s^{n-1}\; =\; \zeta_0 Z_0^{-1}\phi (Z_0s)\,.
\eeq
Then equations\ (\ref{x3}) and (\ref{x6}) can be rewritten as 
\beq\label{fcrit}
f(1) = 1\quad\text{and}\quad f'(1)=1\,.
\eeq
 Moreover, the genericity assumption (\ref{genass}) is
equivalent to assuming $f$ to be analytic in a neighbourhood of
the closed unit disk.

If $T$ is a finite tree, let $h(T)$ denote its {\it height},
i.e.\ the maximal height of vertices in $T$. The set of vertices at height $k$ is called the $k$th generation of $T$ and hence $D_k(T)$ is the size of the $k$th generation.
Clearly, $D_1=1$ and $P (\{D_2=n\})=p_n$ where $P (A)$ is the probability of the event $A$.
 Let $f_n(s)$ be the generating function for $D_n$, i.e.,
\beq\label{Genn}
f_n(s)= \sum_{k=0}^\infty P (\{D_n=k\})s^k.
\eeq
Then of course $f_2(s)=f(s)$.  If we assume that $D_n=k$ then the probability that $D_{n+1}=q$ is given by
\beq\label{ProbGn}
P (D_{n+1}=q \,|\, D_n=k)= \sum_{n_1+n_2+\ldots +n_k=q}p_{n_1}p_{n_2}\ldots p_{n_k}\,,
\eeq
so the generating function for $D_{n+1}$ is $f_{n+1}(s)=f(f_{n}(s))$.
By induction it follows that
$f_{n+2}$ is the $n$th iterate of $f$.

Clearly, the average value of $D_n$ with respect to $\mu$ equals $f'_n(1)$.  By \eqref{fcrit} it follows that
$\br D_n\kt_\mu =1$ for all $n$. As a consequence we get that
\beq \label{ballexp}
\langle |B_R|\rangle_\mu = \sum_{n=1}^R\langle D_n\rangle_\mu = R\;.
\eeq

The probability that the GW process dies out, i.e., the tree has finite height, is
given by 
\beq\label{extinction}
P (D_n=0 ~\mbox{\rm for some} ~n) =\lim_{n\to\infty}P (D_n=0)=\lim_{n\to\infty} f_n(0)
\eeq
since $P (D_{n+1}=0 \,|\, D_n=0)=1$.  Since $f(0)<1$ it is easy to see by induction that 
$f_n(0)<1$ for all $n$.   Furthermore, $f'(s)<1$ for $0\leq s<1$ so $f(s)>s$ for $0\leq s<1$.  It follows that
$f_n(0)$ is increasing in $n$ so the limit $\lim_{n\to\infty}f_n(0)=\lambda$ exists.  Clearly 
$f(\lambda )=\lambda$ and we conclude that $\lambda =1$, so the tree has a finite height with probability 1.

Working slightly harder one can show that
\beq\label{kolmogorov}
P (D_n>0)= \frac{2}{nf''(1)} + O(n^{-2}),
\eeq
if $f''(1)$ is finite.
This means that if $\mu$ is the measure on finite trees 
given by (\ref{mucrit}) then 
\beq \label{heightdistr}
\mu(\{T\in \cT :\; h(T)\geq R\}) = \frac{2}{f''(1)R} + O(R^{-2})
\eeq 
for $R$ large. The proof of (\ref{kolmogorov}) can be found, e.g., in
\cite{Harris}.   

In the special case $p_n=bc^{n-1}$ for $n\geq 1$ and
$p_0=1-\sum_{n\geq 1}p_n$,
the proof is simple since $b=(1-c)^2$ and
\beq\label{Explf}
f(s)=\frac{c+(1-2c)s}{ 1-cs}.
\eeq
The iterates of $f$ can be calculated explicitly:
\beq\label{closedform}
f_{n+1}(s)=\frac{(n+1)c-(nc+2c-1)s}{ 1+nc-(n+1)cs}
\eeq
and
\beq\label{kol2}
1-f_{n+1}(0)=\frac{1-c}{ 1+nc}.
\eeq

\section{Causal triangulations}\label{sec:bijection}

\subsection{Definition}
\label{subsec:CDT}

Let $G$ be a finite rooted planar triangulation with the topology of a disk, i.e.\ a finite planar graph with a root $S_0$ such that all the faces are triangles, except one, called the \emph{exterior face}, whose complement is a closed disk. We say that $G$ is a \emph{causal triangulation} (CT) if the vertices at distance $k$ from $S_0$ span a cycle, i.e. the edges of $S_k(G)$ form a cycle and there are no isolated vertices in $S_k(G)$, for  $0<k<h(G)$, where 
\begin{equation}\label{def:dh}
h(G) = \underset{v\in V(G)}{\mbox{\rm max}} d_G(S_0,v)
\end{equation}
is called the \emph{height} or \emph{radius} of $G$. Thus, for $0<k<h(G)$, each vertex $v$ in $S_k(G)$ has two neighbours in $S_k(G)$ and a number $\sigma_{fv}\geq 1$ of \emph{forward neighbours} in $S_{k+1}(G)$ as well as a number $\sigma_{bv}\geq 1$ of \emph{of backwards neighbours} in $S_{k-1}(G)$ such that 
\beq 
\sigma_v= \sigma_{fv} + \sigma_{bv} + 2\,.
\eeq 
We call $\sigma_{fv}$ the \emph{forward degree} of $v$ and $\sigma_{bv}$ the \emph{backward degree} of $v$, and we say $v$ has height $k$ in $G$ if $v\in S_k(G)$. By convention, we shall assume that each vertex at height $h(G)$ is contained in a single triangle, thus having degree $2$, and implying that boundary vertices of the disk alternate in height between $h(G)$ and $h(G)-1$, see Fig. \ref{fig:0}.  This boundary condition is not essential to the definition of CTs but it is convenient when we come to defining the probability distributions on finite CTs below. 

The preceding definition of finite CTs extends in a straightforward way to the case of infinite CTs, in which case $h(G)=\infty$ and no boundary is present. It is clear that any infinite CT can be drawn in such a way that it covers the whole plane, which we will generally assume in the following. Likewise, the definition above can easily be adapted to CTs with the topology of a cylinder that will be of interest in Section \ref{sec:gce}. 

We denote by $\Cfinite$ the collection of all finite causal triangulations of the disk; by $\C C_\infty$ the set of infinite triangulations of the plane; and by $\C C$ their union. 
 Moreover, let ${\C C}^{(h)}$ be the set consisting of CTs of height $h$. For technical
reasons that will become clear below we will always assume that one of the
edges emerging from the \emph{central vertex} $S_0$ is marked and called the
\emph{root edge}. In particular, this eliminates accidental symmetries under rotations
around the root vertex. An example of $G\in{\C C}^4$ is shown in  Fig. \ref{fig:0}. 
 
 We will consider two different types of ensembles of causal triangulations. In the next section the  \emph{grand canonical ensemble} based on finite CTs will be defined and associated correlation functions calculated. In the present section our main focus is on infinite CTs, making use of the results about ensembles of infinite trees in the previous section via a bijection between planar trees and CTs that we now  describe. 
 
\begin{figure}
\sidecaption[t]
\includegraphics[scale=0.4,angle=-00]{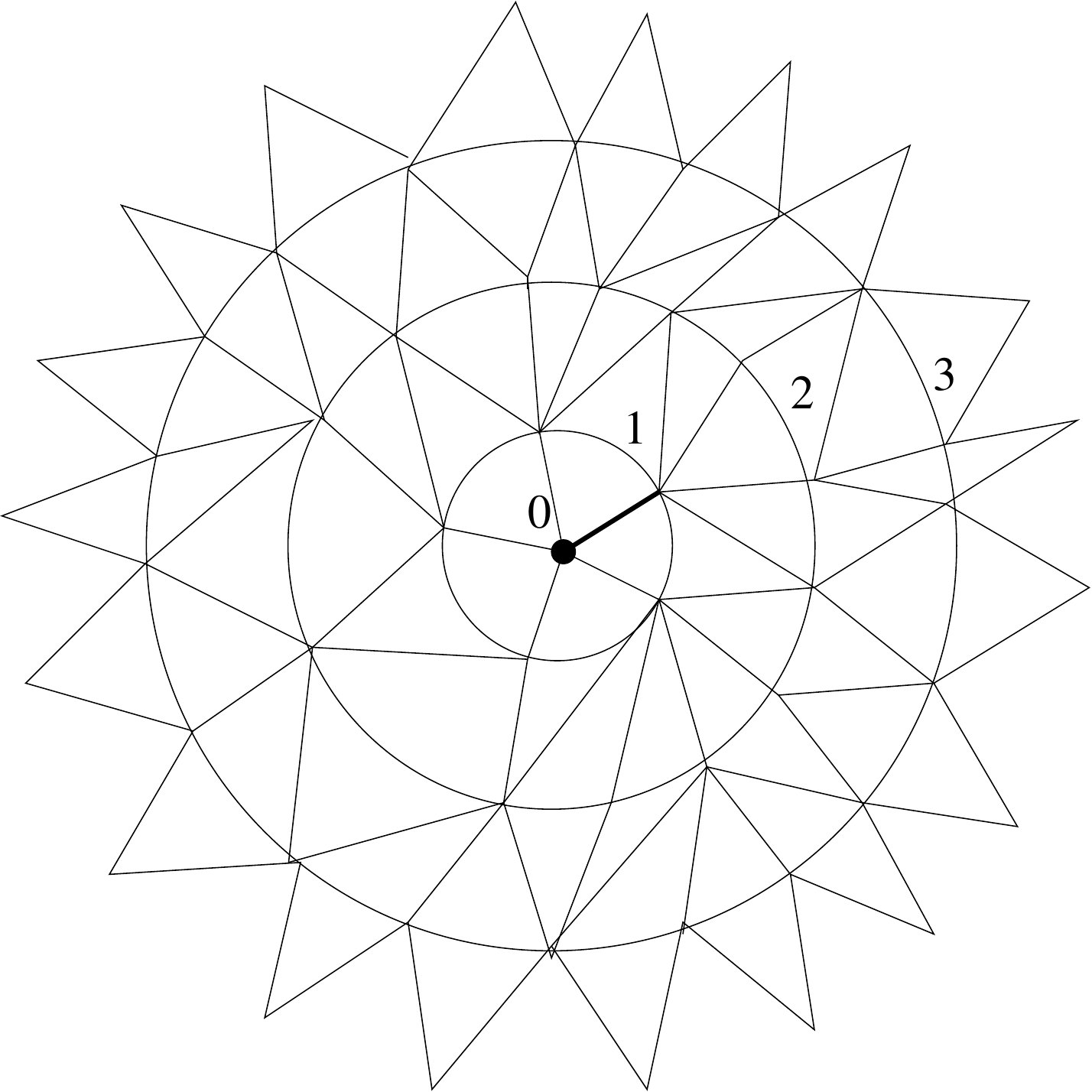}
\caption{Example of $G\in{\C C^{(4)}}$; the numerical labels show the heights of the cycles and the root and marked edge are shown in bold.}

\label{fig:0}       
\end{figure}

\subsection{Bijection between CTs and planar trees}
\label{subsec:bijection}

Given a causal triangulation $G$ and $k<h(G)-1$ we will let $\Sigma_k(G)$ denote the
 subgraph of $G$ spanned by  $S_k(G)$  and $S_{k+1}(G)$, i.e. it consists of the vertices in $S_k(G)$ and $S_{k+1}(G)$ together with the edges joining them. 
 Note that $\Sigma_k(G)$ is a triangulation of an annulus for $k>0$. Denoting by $\Delta(H)$ the number of triangles in a planar graph $H$, we have that 
 \beq \label{sigma} \Delta(\Sigma_k) =\abs{S_k}+\abs{S_{k+1}}\eeq
 and hence, due to the chosen boundary condition, we have for $G\in\Cfinite$ that 
\beq\label{area}
\Delta(G) = 2\,\sum_{k=1}^{h(G)-1}\abs{S_k(G)}\;.
\eeq
In particular, it follows that $\Delta(G)$, which will be called the \emph{area} of $G$, is even. We shall denote by $\C C_N$ the subset of $\Cfinite$ consisting of CTs of area $2N$ for $N\geq 1$. 
 
Let $G\in\Cfinite$. We define a planar rooted tree $T=\beta(G)$ from $G$ in the following way:
\begin{enumerate}
\item The vertices of $T$ are those of $G$ whose  height is at most $h(G)-1$ 
  together with a new vertex $r$ which is the root of $T$ and whose
  only neighbour is $S_0$ and which is placed in the triangle incident on the marked edge on the right as seen from $S_0$.  
 \item All edges in the cycles $S_k(G),\, k=1,2,\dots,h(G)-1$, as well as those containing a vertex at maximal height are deleted, while all edges from $S_0$ to $S_1$ belong to $T$.
\item For each $2\leq k < h(G)-1$ and each vertex $v\in S_k(G)$ the rightmost of the $\sigma_{fv}$ forward pointing edges as seen from $v$ is deleted.
\end{enumerate}
Fig.\ref{fig:CTdef} shows an example of the application of these rules. 
Note that if the height of a vertex in $G$ is $k$ then its height in
$\beta(G)$ is $k+1$, i.e. the vertices in $ D_{k+1}(T)$ coincide with those of $S_k(G)$, $ 0\leq k <h(G)-1$.

Conversely, let $T$ be a rooted planar tree. Then the inverse image
$G=\beta^{-1}(T)$ is obtained as follows: 
\begin{enumerate}
\item Mark the rightmost edge connecting $S_1(T)$ and $S_2(T)$. Delete the root of $T$ and the edge joining it to $S_1(T)$. The
  remaining vertices and edges of $T$  all belong to $G$ and $S_1(T)$
  becomes $S_{0}$, the root of $G$.
\item For $2\leq k \leq h(T)$ insert edges joining vertices in
  $S_{k}(T)$ in the cyclic order determined by the planarity of $T$; this creates the sub-graphs  $S_{k-1}(G)$.
  \footnote{Note that by this convention we allow certain degenerate causal triangulations with cycles $S_k$ 
  having 
  one or two edges corresponding to trees with one or two vertices at a given height.}
\item For every vertex $v\in S_k(T)$, $ 2\leq k\leq h(T)-1$, in $T$ draw an edge from $v$ to a vertex in $S_{k+1}(T)$
 such that the new edge is the rightmost as seen from $v$ to
 $S_{k+1}(T)$ and does not cross any existing edges.
\item Decorate the edges of the cycle $S_{h(T)-1}(G)$ with triangles.
\end{enumerate}

A mapping equivalent to $\beta$ is described in \cite{Krikun:2008}. For $G\in\Cfinite$ these mappings are variants of Schaeffer's bijection
\cite{schaeffer:1998}. Indeed, deleting the edges in $S_k(G)$ for all $k$ and identifying the vertices of maximal height $h(G)$ one
obtains a quadrangulation to which Schaeffer's bijection can be
applied; here the labelling of the vertices equals the height
function. It is clear, that the bijection just described extends to the case of infinite CTs and planar trees, simply by ignoring the points pertaining to the chosen boundary condition for CTs. For an extension to more general
planar quadrangulations see \cite{chassaing:2006}.

\begin{figure}
\sidecaption[t]
  \includegraphics[scale=0.4]{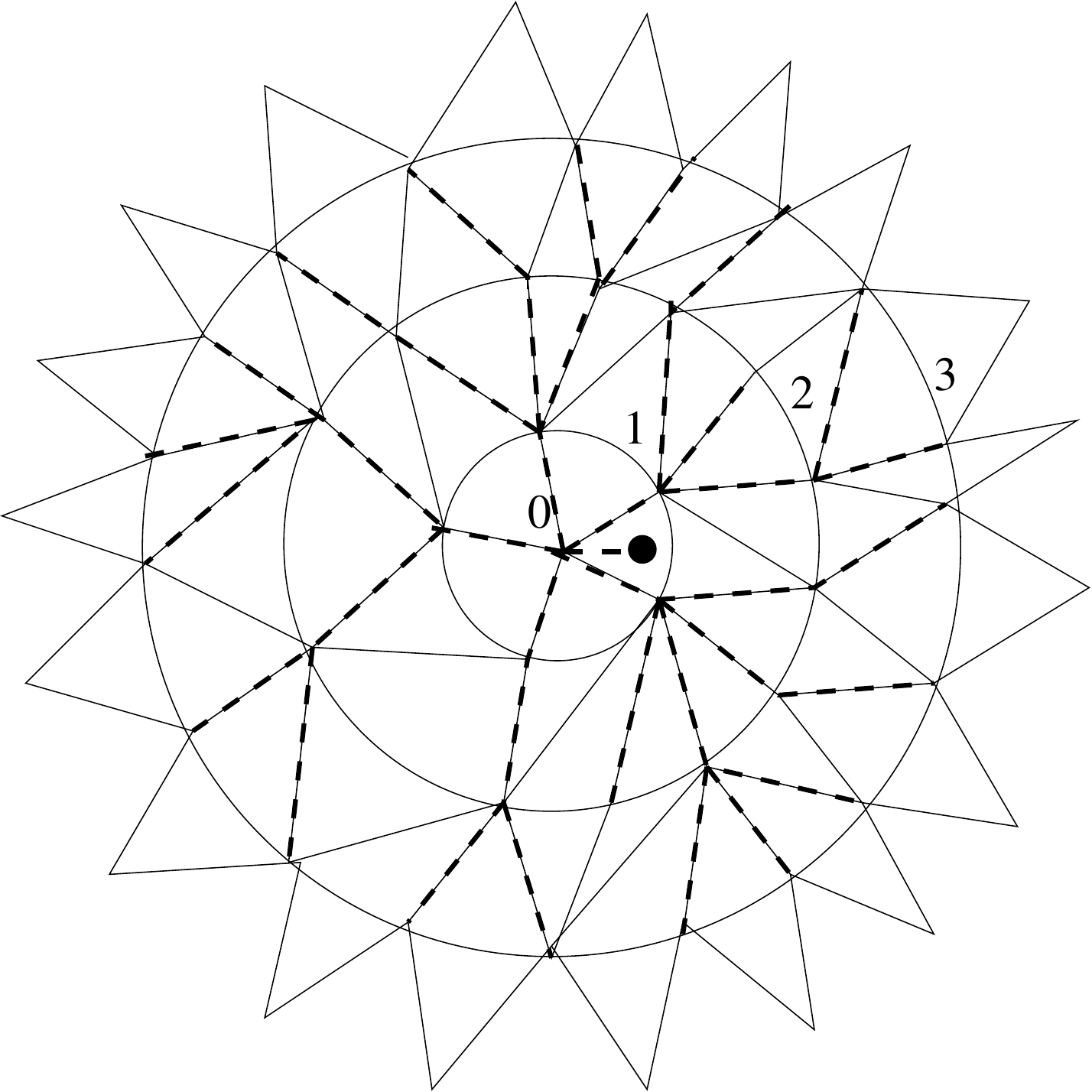}
\caption{The bijection from $G\in{\C C}$ to $T\in{\C T}$: this example shows the tree equivalent to the triangulation of Fig.\ref{fig:0}. The dashed lines show the edges of the tree, including the new edge $(r,S_0)$.}
\label{fig:CTdef}       
\end{figure}

Using \eqref{area}, this construction of $\beta$ shows that it maps ${\cal C}_N$ bijectively onto ${\cal T}_{N+1}$ and likewise $\cal C_\infty$ onto $\cal T_\infty$.
Moreover, defining the metric $d_{\cal C}$ on $\cal C$ by
\beq
d_{\cal C}(G,G') = \inf\{\frac{1}{R+1}\,:\; B_R(G)=B_R(G')\}\,,
\eeq
the map  is an isometry.

Now define the uniform finite volume probability distributions $\rho_N$ by 
\beq\label{tildemuN}
\rho_N(G) =\frac{1}{\sharp\, {\C C}_N} = \frac{1}{\sharp\, {\C T}_{N+1}}\,,\quad G\in {\mathcal{C} }_N\,.
\eeq
 Thus, we have that $\rho_N$ is related to the uniform tree $\nu_N$ (see  \eqref{Nuniftree}) by 
\beq \label{equivN}
\rho_N(G) = \nu_N(\beta(G))\,,\quad G\in {\cal C}_N\,.
\eeq
It follows immediately from the existence of the UIPT $\nu$ discussed in Section \ref{sec:gentree} that the limit $\rho=\lim_{N\to\infty}\rho_N$ exists and is a probability
measure on ${\cal C}_\infty$ given by
\beq \label{equivinfty}
\rho(A) = \nu(\beta(A))
\eeq
for any event $A\subseteq {\cal C}_\infty$.

We call the ensemble $({\C C_\infty},\rho)$ the \emph{Uniform Infinite
  Causal Triangulation} (UICT).
As noted in Section \ref{sec:tree} the measure $\nu$ is concentrated on the set $\C S$ of single spine trees. Hence, $\rho$ is concentrated on the
subset $\beta^{-1}({\cal S})$. 

A result analogous to the above has been obtained for general planar triangulations in \cite{angel:2003aa}. Finally, we observe that the relationship
between trees and CTs described here is not the same as that introduced in \cite{DiFrancesco:1999em} 
where the trees do not in general belong to a generic random tree ensemble.

\section{Grand canonical ensemble and the scaling limit \label{sec:gce}}

\subsection{Disk and annulus partition functions }\label{subsec:partitionfn}

The grand canonical ensemble for finite CTs 
was introduced in \cite{Ambjorn:1998xu}.
The disk partition function  $W_M$ for CTs of a fixed height $h$ is defined by assigning each triangle in $G\in{\C C}^{(h)}$ (see Fig. \ref{fig:0}) a weight $g$, and  each boundary triangle an additional weight factor $yg^{-1}$; this gives

\begin{align}
    \label{eqn:Wdefn}
    W_M(g,z;h)&=\sum_{G\in {\C C}^{(h)}} \abs{S_{h-1}(G)}\, (z/g)^{\abs{S_{h-1}(G)}}\,  g^{\Delta(G)}\nonumber\\
    &=\lderiv{z} \sum_{G\in {\C C}^{(h)}} 
    (z/g)^{\abs{S_{h-1}(G)}}\,  g^{\Delta(G)}.
\end{align}
Here the subscript $M$ indicates that the disk boundary is marked (recall that there is also a marked root edge), which generates the factor $\abs{S_{h-1}(G)}$ in the weight of $G$.
$W_M$  can be thought of as the discretized path integral for the amplitude that a Euclidean universe with disk topology starts at a point $S_0$ at Euclidean time $0$, and has a single connected boundary at Euclidean time $h$. Then $\log g$ is the bulk cosmological constant coupled to $\Delta(G)$ which is the space time volume, and $\log y$ is the boundary cosmological constant coupled to the boundary length given by the number of boundary triangles, $\abs{S_{h-1}(G)}$. 

Correspondingly, the annulus (or cylinder) amplitude describes a Euclidean universe that evolves in Euclidean time $h$ from an entrance boundary to an exit boundary.  The contributing graphs are created from the disk graphs by inserting a second boundary at height 0; 
starting with $G\in{\C C}^{(h)}$ separate the triangles in $\Sigma_0$ so that they no longer have edges in common, but still have an edge in $S_1$. 
The resulting entrance boundary contains $\abs{S_1}$ triangles. Each is assigned an extra weight factor $xg^{-1}$, and one, defined to be the triangle immediately clockwise of the marked edge in $G$, is marked. 
The annulus partition function with 
one marked triangle on the exit boundary is then
\begin{align}\label{eqn:Gdefn}
    W_{MM}(g,x,y;h)&=\sum_{G\in {\C C}^{(h)}} 
    (x/g)^{\abs{S_1(G)}}\, \abs{S_{h-1}(G)}\, (y/g)^{\abs{S_{h-1}(G)}}\,  g^{\Delta(G)}\, .
    \end{align}

The partition functions are computed using the
bijective map $\beta: {\C C}^{(h)}\rightarrow {\C T}^{(h)}$ rather easily.  
Let $w_h(g,z)$ be the partition function for trees of height $\le h$, with each vertex $v$ assigned a weight  $g^{2(\sigma_v-1)}$, and each vertex at height $h$ assigned a further weight $zg^{-1}$, then
\begin{align}
    w_h(g,z)=\sum_{h'\le h}\,\sum_{T\in\C T^{(h')}}(z/g)^{\absd{S_{h}(T)}} \left(\prod_{v\in T \backslash r}
    g^{2(\sigma_v-1)}\right).\label{eqn:wdefn}
\end{align}
Each vertex in $S_{i+1}(T)$ has exactly one edge connecting it to a vertex in $S_{i}(T)$ for $i=1\ldots h(T)-1$. So every vertex $v\in T\backslash r$ contributes $\sigma_v-1$ vertices in $G\backslash S_{h(G)}(G)$, 
where $G=\beta^{-1}(T)$, and thus, using \eqref{area},
\begin{align}
    \Delta(G)=2\sum_{v\in T \backslash r}(\sigma_v-1).
\end{align}
Using the map $\beta$ to rewrite the right hand side of \eqref{eqn:wdefn} as a sum over  CTs  gives
\begin{align}
    w_h(g,z)=\sum_{h'\le h}\,\sum_{G\in{\C C}^{(h')}} (z/g)^{\abs{S_{h-1}(G)}} g^{{\Delta(G)}}.\label{eqn:wdefn2}
\end{align}
Only trees $T$ of height $h$ generate $z$-dependent contributions to \eqref{eqn:wdefn}, so differentiating the right hand side of 
\eqref{eqn:wdefn} w.r.t. $z$ suppresses all except the $h'=h$ terms; hence
\begin{align}
    \lderiv{z}w_h(g,z)=W_M(g,z;h).
\end{align}

\begin{figure}
\sidecaption[t]
  \includegraphics[scale=1.0]{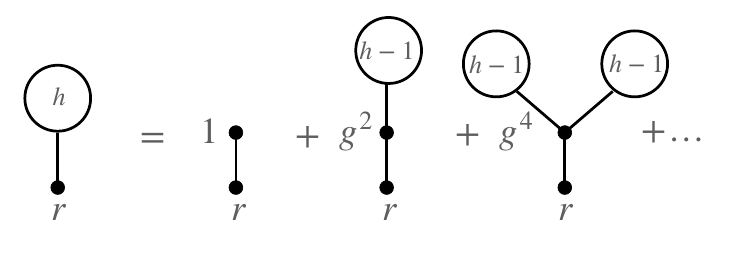}
\caption{Graphical representation of equation \eqref{eqn:wrecurr}.}
\label{fig:treecutting}       
\end{figure}

To compute $w_h(g,z)$ we decompose the trees of height $h$ into trees of height $h-1$ by cutting at the vertex adjacent to the root, see  Fig. \ref{fig:treecutting},
which gives 
\begin{eqnarray}\label{eqn:wrecurr}
    w_h(g,z)&=&\sum_{k=0}^\infty g^{2k}(w_{h-1}(g,z))^k  = \frac{1}{1-g^2 w_{h-1}(g,z)},
\end{eqnarray}
with $w_1(g,z)=zg^{-1}$.
This recursion  is easily solved by setting $w_h=u_h/u_{h+1}$ which gives a linear difference equation for $u_h$; imposing the initial condition, and 
choosing the convenient parametrization $g^{-1}=2\cosh\theta$, leads to
\begin{equation}
    w_h(g,z)=2\cosh\theta\,\frac{\sinh (h-1)\theta-z\sinh(h-2)\theta}{\sinh h\theta-z\sinh (h-1)\theta},
\end{equation}
and hence the disk partition function is 
\begin{align}
    W_M(g,z;h)&=2\cosh\theta \,\frac{z\sinh^2\theta}{(\sinh h\theta-z\sinh (h-1)\theta)^2}.\label{eqn:Whresult}
\end{align}
To find the annulus partition function we follow the same steps until the final iteration of the tree recurrence. Here each offspring of the vertex adjacent to the root has a weight $g^2 (x/g)=xg$ instead of $g^2$, so
\begin{align}\label{eqn:Ghresult}
    W_{MM}(g,x,y;h)&=
    \lderiv{y}
    \frac{1}{1-x g w_{h-1}(g,y)}\nonumber\\
    &=\frac{xy\sinh^2\theta}{(\sinh (h-1)\theta-(x+y)\sinh (h-2)\theta+xy\sinh(h-3)\theta)^2}.
\end{align}
The partition functions \eqref{eqn:Whresult}, \eqref{eqn:Ghresult} for every $h$ are analytic functions of $g,x$, and $y$ in the region 
\begin{align}
    A: \abs{g}<g_c=\half,\;\abs{x}<1,\;\abs{y}<1.\label{eqn:Adefn}
\end{align}
Note that for a given finite $h$ and $g<\half$ the poles of $W_{MM}$ in $x$ and $y$ lie strictly outside $A$.

Finally, we remark that 
$\half w_n(g=\half,s)$ is the offspring probability generating function  $f_{n-1}(s)$ for the uniform random tree, given by \eqref{closedform} with $c=\half$.

\subsection{Scaling amplitudes\label{subsec:scaling}}
As noted above, the partition functions \eqref{eqn:Whresult}, \eqref{eqn:Ghresult} are analytic in the region $A$. 
Within $A$ the partition functions are dominated by graphs with small area and short boundaries. Approaching the limits of $A$, the area and boundary length(s) of the dominant graphs grow arbitrarily large,  
and the scaling limit can be constructed. 
Expanding \eqref{eqn:Whresult} about $\theta =0$ at fixed $h$ and $y<1$ gives
\begin{align}
    W_M(g,y,h)=2\,\frac{y}{(h (1-y)+y)^2}
   +\C O\left(\theta ^2\right),\label{eqn:Wupper}
\end{align}
which reflects the fact that tall trees are rare even at $g=\half$.
The universe described by $W_M$ does not survive when $h\to \infty$ unless the limit is taken in such a way that $h(1-y)=\mathrm{const}$; only then does the model generate universes which are very large, compared to the discretization scale, in  the Euclidean time direction.  The physically non-trivial limit is obtained by setting $g=\half\mathrm {sech}\,\theta$, $y=1-Y\theta\Lambda^{-\half} $, $h=H\theta^{-1}\Lambda^{\half}$ and taking $\theta\to0$.\footnote{A mathematical treatment of the weak convergence properties of this limit is given in \cite{Sisko:2018wpm}.}
The scaling amplitudes are then defined to be
\begin{align}
   W_M^s(\Lambda,Y,H)&= \lim_{\theta\to 0} W_M(\half\,\mathrm {sech}\,\theta,1-Y\theta\Lambda^{-\half},H\theta^{-1}\Lambda^{\half})\nonumber\\
   &=2\,\frac{\Lambda}{(\Lambda^\half\cosh H\Lambda^\half +Y\sinh H\Lambda^\half)^2},\label{eqn:Ws}
\end{align}
and
\begin{align}
   W_{MM}^s(\Lambda,X,Y,H)&= \lim_{\theta\to 0} \theta^2 W_{MM}(\half\,\mathrm {sech}\,\theta,1-X\theta\Lambda^{-\half},1-Y\theta\Lambda^{-\half},H\theta^{-1}\Lambda^{\half})\hfill\nonumber\\
   &= \frac{\Lambda }{\left((\Lambda+XY)  \sinh \asqrt{\Lambda } H+\asqrt{\Lambda } (X+Y)
   \cosh \asqrt{\Lambda } H 
   \right)^2}. \label{eqn:Gs}
\end{align}
The pre-factor $\theta^2$ in the definition of $W_{MM}^s$ reflects the insertion of an extra marked boundary relative to $W_M^s$ which renders the partition function divergent at the critical point.
 $W_M^s$ is the amplitude for a continuum Euclidean universe with disk topology starting at Euclidean time $0$ and having a boundary at Euclidean time $H$ with bulk cosmological constant $\Lambda$ and  boundary cosmological constant $Y$; similarly $W_{MM}^s$ describes a universe with, in addition, a boundary at time $0$ having boundary cosmological constant $X$. $H$ is chosen to have length dimension $[H]=1$, so $[\Lambda]=-2$, and the extents of the boundaries conjugate to $X$ and $Y$ also have dimension 1. The scaling dimension $d_H$ (sometimes called the scaling Hausdorff dimension) is defined through the dependence of the average area of graphs $\langle \Delta(G)\rangle_{\theta}$  on the height $h(\theta)=H\theta^{-1}\Lambda^{\half}$ in the scaling limit of the disk ensemble  
 \begin{align}
        d_H&= \lim_{\theta\to 0}\frac{\log \langle \Delta(G)\rangle_{\theta}}{\log h(\theta)},
 \end{align}
 where
 \begin{align}
      \langle \Delta(G)\rangle_{\theta}=\left. \lderiv{g}\log W_M(g,y(\theta);h(\theta))\right\vert_{g=\half\mathrm {sech}\,\theta,\,y(\theta)=1-Y\theta\Lambda^{-\half}}.
 \end{align}
 This gives $d_H=2$ which is consistent with the dimension of the spatial and temporal extents each being 1, and  the universes described by the scaling limit being colloquially two-dimensional.  See \cite{Zohren:2008vqi} for further discussion  of these partition functions.

\section{Hausdorff dimension}\label{sec:hausdorff}

In the previous section we introduced the dimension $d_H$ which relates the total area and the linear extent in the limit when both become large.
   This section takes another point of view; we consider infinite graphs 
and the relation between the size of a ball and its radius  as the latter becomes large.

The Hausdorff dimension $d_h$ (sometimes called the local Hausdorff dimension) of an infinite rooted graph $G$ is defined by the 
relation
\beq\label{Hausdorffdim}
|B_R(G)|\sim R^{d_h},~~R\to\infty,
\eeq
where $B_R$ as usual denotes the ball of radius $R$ around the root and $|B_R|$ its size.
More precisely, we define 
\beq\label{Hausdorffdim2}
d_h=\lim_{R\to\infty}\frac {\ln |B_R(G)|}{ \ln R}
\eeq
whenever the limit exists. For the ensembles of trees and surfaces that are studied
here we show that this is indeed the case and yields the same value of $d_h$ for almost all $G$. We shall likewise see that the same value of $d_h$ is obtained by replacing $|B_R(G)|$ in \eqref{Hausdorffdim2} by its average value, which is in general easier to evaluate or estimate. 

  In many important cases $d_h=d_H$; this includes the ensembles studied in this paper, but the relation does not hold universally as will be seen in Section \ref{sec:matterfields}, albeit in a case where the ensemble weights may take negative values.

\subsection{Generic trees}

Let $T$ be a generic tree with associated probability distribution $\nu$.  We can assume, as has been explained
in Section \ref{sec:tree},
that $T$ has a unique spine.  Let $T_i^n$, $n=0,1,\ldots \sigma_{s_i}-2$ be the 
finite trees attached to the $i$th vertex on the spine, see Fig. \ref{fig1}, and  recall that these are independent and each is distributed according to the probability measure $\mu$ given by \eqref{mucrit}. With notation as in Section \ref{sec:gentree} we have

\beq \label{union}
Br_i=\bigcup_{n=1}^{\sigma_{s_i}-2} T_i^n\,,  
\eeq
interpreted as the empty graph if $\sigma_{s_i} =2$. Letting $Y_{j}$ denote the number of vertices different from $s_i$ in $Br_i$ located at distance 
$\leq j$ from $s_i$ we can write
\beq\label{volball}
|B_R|=R+\sum_{i=1}^R Y_{R-i}\,,
\eeq
where the $R$ on the right hand side
accounts for the number edges on the spine inside $B_R$.  It follows from 
(\ref{branchprob}) and (\ref{GWprobw}) that
\beq\label{spineprob}
\nu(\{\sigma_{s_i}=n+2\})=(n+1)p_{n+1}.
\eeq
When multiplied by $n$ this is the $(n-1)$th Taylor coefficient of $f''$. 
Using (\ref{spineprob}) and (\ref{ballexp}) this gives       
\beq\label{expY}
\br Y_{R-i}\kt_\nu =f''(1)(R-i).
\eeq
Summing over $i$ from 1 to $R$ yields
\beq\label{BRExp}
\br |B_R|\kt_\nu =\oh f''(1)R(R-1)+R,
\eeq
which shows that in terms of average values of ball sizes we have $d_h=2$. 

To obtain bounds on $|B_R(T)|$ for individual trees is more cumbersome and we shall not elaborate in detail on this issue here. In section \ref{sec:5.3}
we show that if $G_R=\cup_{i=1}^R Br_i$ then 
$$
|G_R|\,\leq \, C_2 R^2(\ln R)^3
$$
holds for $R$ large enough almost surely with respect to $\nu$ (see \eqref{boundG0K}). Since we clearly have $B_R\subseteq G_R$ the same bound holds for $|B_R|$.
A similar lower bound is shown in \cite{Durhuus:2009sm} yielding
the a.s.\ bounds
\beq\label{BoundsT}
C_1(\ln R)^{-2}R^2\leq |B_R(T)|\leq C_2R^2(\ln R)^3,
\eeq
where $C_1$ and $C_2$ are positive constants.
Evidently these bounds imply that $d_h =2$ a.s.

It is worth remarking that the ensemble average of the volume of a ball 
$B_R(v)$ centered at a random vertex $v$ within some fixed distance from the root displays the same behaviour as in \eqref{BoundsT} as a simple consequence of the triangle inequality.

\subsection{Causal triangulations}

We now turn to the Hausdorff dimension of causal triangulations.  For an infinite  
causal triangulation $G$ we have 
\beq\label{a1}
|B_R(G)|=2\sum_{i=1}^{R}|S_i|+ \sum_{i=1}^{R-1}|S_i|\,,
\eeq
and it follows that
\beq\label {a2}
\| B_R\|\leq |B_R|\leq 3\| B_R\| .
\eeq
Clearly $\|B_R\|+1$ equals the number
of vertices within distance $R$ from the root of the tree $T$ corresponding to
$G$ under the bijection $\beta$.  Hence, in view of (\ref{BRExp}),
\beq\label{volx}
\br | B_R|\kt_\rho \sim R^2,~~ R\to\infty ,
\eeq
where the expectation is with respect to the measure $\rho$
defined in \eqref{equivinfty}, so $d_h=2$ for CTs.
By the same argument we likewise have a.s. that 
\beq\label{BoundsCT}
C_1^\prime(\ln R)^{-2}R^2\leq |B_R(G)|\leq C_2^\prime R^2(\ln R)^3,
\eeq
and hence that $d_h=2$ a.s. with respect to $\rho$.

\section{Spectral dimension}\label{sec:spectraldim}

In this section we define a notion of dimension for graphs which is
different from the ones  discussed above.
This is the spectral dimension which is a measure of how likely it is that
a random walker returns to the starting point.   In the following subsections
we analyse the relation between the spectral and Hausdorff dimensions and
calculate the spectral dimension for generic trees and causal triangulations.

\subsection{Definition of spectral dimension of recurrent graphs}\label{sec:5.1}

Given a graph $G$, we use the notation $\omega:v\to u$ to indicate a path $\omega$ from vertex $v$ 
to vertex $u$ and, if $\omega$ has length $|\omega|=m$, the vertices of $\omega$ will be denoted by 
$v = \omega_0, \omega_1,\dots,\omega_{m-1},\omega_{m} =u$. If $v\neq u$ we write $\omega:v\twoheadrightarrow u$ 
for a path from $v$ to $u$ that does not return to $v$, i.e. $\omega_i\neq v$ for $i\neq 0$. If $v=u$, 
the notation $\omega:v\twoheadrightarrow v$ is used for a path from $v$ to $v$ that does not return to $v$ in between, i.e. $\omega_i\neq v$ for $i=1,2,\dots,|\omega|-1$.  Below we also consider infinite paths 
$\omega = (\omega_0\omega_1),(\omega_1\omega_2),(\omega_{2}\omega_{3}),\dots$ emerging from a vertex 
$v=\omega_0$.  

We define the function $p_G$ on the set of all finite paths on $G$ by 
$$
p_G(\omega) = \prod_{i=0}^{|\omega|-1} \sigma_{\omega_i}^{-1}\,.
$$
It is easily seen that $p_G$ defines a probability distribution on the set 
$\Pi_m(v)$ of paths of fixed length $m$ and fixed initial vertex $v$ and that these distributions are 
compatible in the  way described in Section\ref{sec:prelim}.   We define a probability distribution $P_{G,v}$ on the set 
$\Pi_\infty(v)$ of all infinite paths starting at $v$ by setting 
$$
P_{G,v}(A(\bar\omega)) = p_G(\bar\omega)\,,  
$$
where $\bar\omega$ is an arbitrary finite path starting at $v$ and $A(\bar\omega)$ denotes the set of all 
infinite paths that coincide with $\bar\omega$ in the first $|\bar\omega|$ edges. When considered as probability spaces in the way described, the paths in $\Pi_m(v)$ or $\Pi_\infty(v)$ are usually referred to as  \emph{random walks}.\footnote{More commonly, they are called simple random walks, to distinguish them from, e.g., biased random walks. Since we do not  consider different kinds of random walks in this paper we will leave out the adjective "simple".}

The probability for a random walk of length $m$ starting at $v$ 
to end at $u$ is given by
$$
q_G(m;v,u) = \sum_{\omega:v\to u, |\omega|=m} p_G(\omega)\,.
$$
 The corresponding \emph{cumulative probability} is defined as
$$
Q_G(n;v,u) = \sum_{m =0}^n q_G(m;v,u)
$$
for any $n=0,1,2,\dots$, with the convention 
$$
q_G(0;v,u) = \delta_{v,u} = \begin{cases} 1\quad\mbox{if $v=u$}\\ 0\quad\mbox{if $v\neq u$}\,,\end{cases}
$$ 
i.e., we define $p_G(\omega)=1$ for the trivial walk $\omega$ of length $0$ consisting of a 
single vertex. 

With this convention we note for later reference that $Q_G$ fulfills  
\begin{equation}\label{heat}
Q_G(n+1;v,u)=  \sum_{x:(u,x)\in E(G)} \sigma_x^{-1}Q_G(n;v,x) + \delta_{v,u}\,,\quad 
n\geq 0\,,
\end{equation}
where the sum on the right-hand side, as indicated, is over 
the neighbours of $u$. The quantity
$$
Q'_G(n;v,u) = \sigma_u^{-1} Q_G(n;v,u)\,,
$$
which is symmetric in $v$ and $u$, then fulfills the discrete version of the diffusion equation with a 
source at vertex $v$:
\begin{equation}\label{diffusioneq}
\partial_n Q'_G(n;v,u) = - \triangle^G_u Q'_G(n;v,u) + \sigma_v^{-1}\delta_{v,u}\,,\quad n\geq 0\,, 
\end{equation}
as is easily seen by subtracting $Q_G(n;v,u)$ from  
both sides of equation \eqref{heat}.
Here $\partial_n$ denotes the difference operator with respect to "time" $n$ and $\triangle^G$ is the 
graph Laplace operator acting on functions $f:V(G)\to \mathbb C$ according to
$$
\triangle^Gf(v) = \sigma_v^{-1}\sum_{x:(v,x)\in E(G)} (f(v)-f(x))\,. 
$$

\smallskip

The \emph{spectral dimension} of a connected graph $G$ is most commonly defined in terms of the decay rate of 
the return probability $q_G(m;v,v)$ as a function of $m$. More precisely, if 
\beq\label{specdim1}
q_G(m;v,v) \sim m^{-\frac{\alpha}{2}}\quad\mbox{for $m$ large,}
\eeq
we call $\alpha$ the spectral dimension of $G$ and in this case $G$ is called {\it recurrent} if $\alpha\leq 2$, and otherwise it is called {\it transient}.  More generally, noting that $Q_G(n;v,v)$ is always an increasing function of $n$, the limit $Q_G(\infty;v,v) :=\lim_{n\to\infty} Q_G(n;v,v)$ exists, and $G$ is recurrent if 
the limit is $\infty$, otherwise $G$ is transient. Furthermore, we find it most convenient for our purposes to define the spectral dimension in terms of the asymptotic behaviour of $Q_G(n;v,v)$ for large $n$.
 Thus, for a recurrent graph $G$, we set 
\begin{equation}\label{eq:s-dim}
d_s = 2 - 2\lim_{n\to\infty} \frac{\ln Q_G(n;v,v)}{\ln n}\,,
\end{equation}
provided the limit exists (in which case its value is independent of $v$).  The definition (\ref{eq:s-dim}) is equivalent to (\ref{specdim1})
under mild assumptions.

Note that, since  
$1\leq Q_G(n;v,v)\leq n$ we have  $0\leq d_s \leq 2$. Obviously, $Q_G(n;v,v)$ is not a probability, contrary to $q_G(n;v,v)$. On the other hand, letting $q_G^0(m;v,v)$ denote the first return probability after $m$ steps of the walk, i.e.
$$
q_G^0(m;v,v) = \sum_{\displaystyle 
\omega:v \twoheadrightarrow v, |\omega|=m} p_G(\omega)\,,
$$
we have that
$$
Q_G^0(n;v,v) = \sum_{m=2}^n q_G^0(m;v,v)
$$
is the probability that an infinite walk starting at $v$ returns to $v$ after at most $n$ steps and, in particular, 
$$
Q_G^0(\infty;v,v) = \lim_{n\to\infty} Q_G^0(n;v,v)
$$
is the probability that the infinite random walk returns at least once to $v$. Denoting by $\Pi_\infty^m(v)$ the set of walks that return to $v$ at least $m$ times we can decompose each such walk $\omega$ into pieces  $\omega^{(1)},\dots,\omega^{(m)},\bar\omega$ such that $\omega^{(k)}; v\twoheadrightarrow v$ for $k=1,2,\dots, m$ while $\bar\omega\in\Pi_\infty(v)$ is an end piece. Then the set $A(\omega^{(1)},\dots,\omega^{(m)})$ of all paths in $\Pi_\infty(v)$ whose decomposition is of the stated form with fixed $\omega_1,\dots,\omega_m$ and arbitrary $\bar\omega$ has probability 
$$
P_{G,v}(A(\omega^{(1)},\dots,\omega^{(m)})) = P_{G,v}(A(\omega^{(1)}))\cdots P_{G,v}(A(\omega^{(m)}))\,.
$$
By summing over $\omega^{(1)},\dots,\omega^{(m)}$ we obtain that the probability that the random walk returns at least $m$ times to $v$ equals $Q_G^0(\infty;v,v)^m$\,. Letting $m$ tend to infinity we conclude that the probability that the random walk returns infinitely many times to the initial vertex vanishes if and only if  $Q_G^0(\infty;v,v) < 1$ and the relation 
$$
Q_G(\infty;v,v) = \frac{1}{1-Q_G^0(\infty;v,v)}
$$
holds. On the other hand, $G$ is recurrent if and only if $Q_G^0(\infty;v,v)=1$ and in that case the walk returns to the initial vertex infinitely many times with probability $1$. 
 It is well known, and easy to see, that if $G$ is finite then $d_s=0$ while if $G$ is the hypercubic lattice $\mathbb Z^d$ (viewed as a graph in the standard way) it is a classical result of Polya, see, e.g., Ch.2 in  \cite{LyoPer}, 
 that $G$ is recurrent if and only if $d\leq 2$, and in all cases $d_s=d$. In this article we are mainly 
 concerned with recurrent graphs.

\subsection{Relation between $d_s$ and $d_h$}\label{sec:5.2}

 We now give an elementary proof of a well known inequality between the spectral dimension and 
the Hausdorff dimension $d_h$ valid for arbitrary recurrent graphs. 
This inequality has been proven under certain assumptions on 
the behaviour of the volume of balls under scaling in \cite{grigoryan1998random,coulhon}. Related results for Riemannian  
manifolds were obtained earlier under similar assumptions in \cite{grigoryan1992heat}. Here we essentially need no assumptions beyond existence of $d_s$ and $d_h$. 
Specifically, we now show that if $G$ is a connected recurrent graph such that $d_s$ and $d_h$ both exist, then
\begin{equation}\label{dimineq} 
d_s \geq \frac{2d_h}{1+d_h}\,.
\end{equation}

The proof is based on a simple observation whose formulation requires some further notation. Thus, let $G_0$ be a subgraph of a 
graph $G$. The \emph{inner boundary} $\partial_{in}G_0$  of $G_0$ is the subgraph of $G$ spanned by the vertices of $G_0$ having at least one neighbour in $V(G)\setminus V(G_0)$. Similarly, the 
\emph{outer boundary} $\partial_{out}G_0$ is the subgraph of $G$ spanned by the vertices not in $G_0$ having at 
least one neighbour in $G_0$. The closure $\bar G_0$ of $G_0$ is the subgraph spanned by the vertices of 
$G_0$ and those of  $\partial_{out}G_0$. The \emph{out-degree} $\sigma_v^{out}$ of a vertex $v$ in $G_0$ is by 
definition the number of neighbours of $v$ in $G$ that do not belong to $G_0$. In particular, a vertex 
of $G_0$ belongs to $\partial_{in}G_0$ if and only if its out-degree is positive.

Now, let $G$ be a connected graph,  $V_0$ a proper subset of $V(G)$, and denote by $G_0$ the subgraph of $G$ spanned by $V_0$. Then, for arbitrary fixed $v_0\in V_0$, we have 
\begin{equation}\label{lemma1}
\sum_{v\in \partial_{in}G_0}\sum_{\omega: \scriptstyle v_0\to v\atop\scriptstyle\omega\subseteq G_0} p_G(\omega)\sigma_v^{-1}\sigma_v^{out} \leq 1\,,
\end{equation}
with equality holding if $\bar G_0$ is connected and recurrent.

 The inequality \eqref{lemma1} follows by observing that the left-hand side is the probability $q$ with 
respect to $P_{G,v_0}$ that 
a walk starting at $v_0$ leaves $G_0$. In fact, given such a walk $\tilde\omega$ let $v$ denote 
the last vertex in $G_0$ visited by $\tilde\omega$ before it leaves $G_0$ for the first time and let $\omega$ denote the corresponding initial part of $\tilde\omega$ from $v_0$ to $v$ contained in $G_0$. Then $v\in\partial_{in}G_0$ and there are $\sigma_v^{out}$ vertices in $\bar G_0$ that $\tilde\omega$ may hit next with each such possibility contributing a probability $p_G(\omega)\sigma_v^{-1}$ to $q$. This proves \eqref{lemma1}. Clearly, $q$ only depends on $\bar G_0$ and if this graph is connected and recurrent it is well known \cite{Feller} that the probability for a walk to hit any given 
vertex of $\bar G_0$ equals $1$. In particular, since $V_0\neq V(G)$, it follows that $q=1$.

For use in the proof of \eqref{dimineq} we note two useful consequences of \eqref{lemma1}. First, let $G$ be a connected graph and let $v_0$ and  $u_0$ be two different vertices of $G$. Then
\begin{equation}\label{cor1}
\sum_{\omega: \scriptstyle v_0\to v_0 \atop 
\scriptstyle u_0\in\omega} p_G(\omega) \leq \sigma_{v_0}\sigma_{u_0}^{-1} \sum_{\omega:v_0\to u_0} p_G(\omega)\,,
\end{equation}
and equality holds if $G$ is recurrent.

To prove this statement we set $V_0= V(G)\setminus \{u_0\}$ and let $G_0$ be the sub-graph of $G$ spanned by $V_0$. In other 
words, $G_0$ is obtained from $G$ by removing $u_0$ and the edges containing $u_0$, and is frequently denoted 
by $G-u_0$.
Now note that $v_0\in G_0$ and that every walk $\omega:v_0\to v_0$ containing $u_0$ can be decomposed uniquely 
into a walk $\omega':v_0\to u_0$ and 
a walk $\omega'':u_0\to v_0$, such that $\omega''$ does not return to $u_0$. 
Hence, the reverse of $\omega''$ 
is a walk $\omega'''$ in $G_0$ from $v_0$ to some $v\in \partial_{in} G_0$ and one additional step to $u_0$. Since 
$\sigma^{out}_v =1$ for all $v\in\partial_{in} G_0$ in this case, (\ref{lemma1}) gives 
\begin{eqnarray*}
\sum_{\scriptstyle \omega: v_0\to v_0 \atop 
\scriptstyle u_0\in\omega} p_G(\omega) &=&   \sum_{\omega':v_0\to u_0}p_G(\omega')\sigma_{u_0}^{-1} 
\sum_{v\in\partial_{in}G_0}\sum_{\scriptstyle 
\omega''':v_0\to v \atop \scriptstyle \omega'''\subseteq G_0} p_G(\omega''')\sigma_v^{-1}\sigma_{v_0} \\
 &\leq&  \sum_{\omega':v_0\to u_0}\sigma_{v_0}\sigma_{u_0}^{-1} p_G(\omega')\,,   
\end{eqnarray*}
with equality holding if $G$ is recurrent. This proves \eqref{cor1}.

Second, with $G$ and $v_0$ and  $u_0$ as above we have that  

\begin{equation}\label{cor2}
\sum_{\omega: v_0\to v_0 \atop u_0\notin\omega} p_G(\omega) \leq \sigma_{v_0} d_G(v_0,u_0)\,.
\end{equation}

\noindent In order to verify this claim,  let $(v_0v_1),(v_1v_2),\dots,(v_{N-2}v_{N-1}),(v_{N-1}u_0)$ be a path from $v_0$ to $w_0$ of  minimal length $N=d_G(v_0,u_0)$.
Set $v_N=u_0$ and define $k_\omega$, for each walk $\omega:v_0\to v_0$, to be the 
maximal index $k$ such that $v_k\in\omega$. In particular, if $u_0\notin\omega$ then  $k_\omega\leq N-1$ and
$$
v_{k_\omega} \in\omega,\;\; v_{k_\omega+1},\dots, v_N\notin\omega\,.
$$  
Given that $k_\omega=l\geq 1$ there is a unique decomposition of $\omega$  
 into a walk $\omega':v_0\to v_l$ and a walk $\omega'':v_l\to v_0$ such that $\omega''$ does not 
return to $v_l$. As previously, the reverse of $\omega''$ is a walk $\omega'''$ from $v_0$ to a 
neighbour $v$ of $v_l$ avoiding $v_l, v_{l+1},\dots, v_N$, and an additional last step from $v$ to $v_l$. 
Setting  $V_0=V(G)\setminus \{v_{l},\dots,v_N\}$ and  $V_1=V(G)\setminus \{v_{l+1},\dots,v_N\}$ and 
letting $G_0$ and $G_1$ be the subgraphs of $G$ spanned by $V_0$ and $V_1$, respectively, it follows that 
$\omega'\subseteq G_1$ and $\omega'''\subseteq G_0$. Noting that $v\in\partial_{in}G_0$ and that 
$$
p_G(\omega) = p_G(\omega')\sigma_{v_l}^{-1}\sigma_{v_0}p_G(\omega''')\sigma_v^{-1}
$$
we conclude that 
\begin{equation}
\sum_{\scriptstyle \omega: v_0\to v_0 \atop 
\scriptstyle k_\omega=l} p_G(\omega) \leq 
\sigma_{v_0}\left(\sum_{\scriptstyle \omega':v_0\to v_l\atop \scriptstyle \omega'\subseteq G_1}p_G(\omega')\sigma_{v_l}^{-1}\right)
\left(\sum_{v\in\partial_{in} G_0}
\sum_{\scriptstyle \omega''':v_0\to v\atop \scriptstyle \omega'''\subseteq G_0}p_G(\omega''')\sigma_v^{-1}\right)\,.
\end{equation}
Since $v_l\in\partial_{in} G_1$ it follows from \eqref{lemma1} that the expressions in parentheses on the 
right-hand side are bounded by $1$ such that 
\beq\label{x500}
\sum_{\omega: v_0\to v_0 \atop k_\omega=l} p_G(\omega) \leq \sigma_{v_0}\,.
\eeq 
In the case $k_\omega =0$ the walk $\omega''$ is trivial and the reader easily verifies that the above inequality still holds  with the trivial walk  also included on the left-hand side contributing $1$ to the sum. 
Finally, summing on both sides of (\ref{x500}) from $l=0$
to $l=N-1$, the claimed inequality \eqref{cor2} follows. 

We are now ready to give a proof of \eqref{dimineq}. Let $v_0\in G$ be fixed. Since $p_G$ is a probability distribution on walks of length $m$ starting at $v_0$, 
the identity 
\begin{equation}\label{sumrule}
\sum_{v\in G}Q_G(n;v_0,v) = \sum_{v\in G}\sum_{\scriptstyle \omega:v_0\to v\atop \scriptstyle |\omega|\leq n}p_G(\omega) = n+1
\end{equation} 
holds for each $n\geq 0$. Restricting the sum on the left-hand side to vertices in $B_R(G;v_0)$ we obtain an inequality instead, which implies 
 that there exists a vertex $v_{R,n}$ in $B_R(G;v_0)$ such that 
\begin{equation}\label{max}
Q_G(n;v_0,v_{R,n}) \leq \frac{n+1}{\Vert B_R(G;v_0)\Vert}\,,
\end{equation}
where $R$ is an arbitrary positive integer. Writing 
$$
Q_G(n;v_0,v_0) = \sum_{\scriptstyle \omega:v_0\to v_0\atop
\scriptstyle v_{R,n}\in\omega, |\omega|\leq n} p_G(\omega) + 
\sum_{\scriptstyle \omega:v_0\to v_0\atop \scriptstyle v_{R,n}\notin\omega, |\omega|\leq n} p_G(\omega)\,,
$$
it follows from \eqref{cor1} and \eqref{cor2} together with \eqref{max} that
\begin{eqnarray}\label{Qupperbound}
Q_G(n;v_0,v_0) &\leq& \sigma_{v_0}\sigma_{v_{R,n}}^{-1}\sum_{\scriptstyle 
\omega:v_0\to v_{R,n}\atop \scriptstyle |\omega|\leq n} p_G(\omega) 
+ \sigma_{v_0}d_G(v_0,v_{R,n})\nonumber\\
&\leq& \sigma_{v_0} Q_G(n;v_0,v_{R,n}) + \sigma_{v_0}d_G(v_0,v_{R,n}) \nonumber\\ 
&\leq&  \sigma_{v_0}\left(\frac{n+1}{\Vert B_R(G;v_0)\Vert} + R\right)\,.
\end{eqnarray} 
Now, choose $R$ as a function of $n$ such that the two terms in parenthesis are of the same order of magnitude. 
This is obtained for  $R=\lfloor n^{\frac{1}{1+d_h}}\rfloor$, where $\lfloor a \rfloor$ denotes the integer part of the real number $a$. 
In this case, the inequality \eqref{Qupperbound} gives 
$$
\frac{\ln Q_G(n;v_0,v_0)}{\ln n} \leq \frac{1}{1+d_h} + \frac{\ln\sigma_{v_0}+ \ln\left(1+\frac{\displaystyle n+1}
{\displaystyle R\Vert B_R(G;v_0)\Vert}\right)}{\ln n}\,.
$$
Here, the last term tends to zero as $n\to \infty$ by the assumption that $d_h$ exists and hence, by \eqref{eq:s-dim}, we get 
$$
d_s\geq 2-2\frac{1}{1+d_h} = \frac{2d_h}{1+d_h}
$$
as desired.

\subsection{Spectral dimension of generic trees}\label{sec:5.3}

For the generic trees we noted in Section \ref{sec:hausdorff} that $d_h=2$  with probability $1$. Hence, it follows from  \eqref{dimineq} that 
\begin{equation}\label{dsineq}
d_s\geq \frac 43
\end{equation}
with probability $1$, provided the limit \eqref{eq:s-dim} exists. We do not provide detailed arguments for the existence of the limit, but some further comments on this issue can be found at the end of this subsection.  Next, we aim at proving that  equality holds in \eqref{dsineq} and for that we need a suitable lower bound on $Q(n;v_0,v_0)$ supplementing the upper bound  \eqref{Qupperbound}. 

In \cite{Durhuus:2007} a rather special type of lower bound on generating functions for return 
probabilities of random walk on generic trees was proven. Here, we establish a natural generalization of that bound applicable to the cumulated probabilities $Q_G(n;v_0,v_0)$ associated with an arbitrary connected graph, as stated in the following theorem.
 
\begin{theorem}\label{thm2}
Let $G$ be a connected graph and $G_0$ a finite connected sub-graph of $G$ spanned by its set of vertices $V_0$. Then
\begin{equation}\label{low2}
Q_G(n;v_0,v_0) \geq \sigma_{v_0}\left(\frac{2|\bar G_0|}{n+1} + {\mathscr C}_{G,G_0}(n;v_0)\right)^{-1}\,,
\end{equation}
for every vertex $v_0\in G_0$, where ${\mathscr C}_{G,G_0}(n;v_0)$ is defined, up to a factor $\sigma_{v_0}$, as the probability for a walk starting at $v_0$ 
not to return to $v_0$ before leaving $G_0$ in at most $n$ steps, that is
\begin{equation}\label{CGG0ndef}
{\mathscr C}_{G,G_0}(n;v_0) = \sigma_{v_0}\sum_{\displaystyle v\in\partial_{in}G_0}\sum_{\displaystyle 
\omega:v_0\twoheadrightarrow v \atop 
\displaystyle \omega\subseteq G_0, |\omega|\leq n} p_G(\omega)\sigma_v^{-1}\sigma_v^{out}\,. 
\end{equation}  
\end{theorem}

\begin{proof}
We define 
$$
\widetilde Q_G(n;v,u) = \sigma_u^{-1}\sum_{\scriptstyle 
\omega:v\to u \atop \scriptstyle \omega\subseteq G_0, |\omega|\leq n}p_G(\omega)
$$
and note that $\widetilde Q_G(n;v,u)$ vanishes if $v\notin V_0$ or $u\notin V_0$, while it satisfies the diffusion 
equation \eqref{diffusioneq} for $v,u\in V_0$. Since the left-hand side of \eqref{diffusioneq} is non-negative (as $\widetilde Q(n;v,u)$ is a non-decreasing function of $n$), it follows from the maximum principle 
for the discrete Laplacian
that $\widetilde Q_G(n;v_0,u)$ assumes its maximal value as a 
function of $u$ at $u=v_0$, i.e.  
\begin{equation}\label{eq.4.4}
\widetilde Q_G(n;v_0,u) \leq  \widetilde Q_G(n;v_0,v_0)\,,\quad u\in G\,.
\end{equation}
From \eqref{sumrule} we have 
\begin{equation}\label{eq.4.5}
n+1 = \sum_{\scriptstyle u\in G_0} \sigma_u \widetilde Q_G(n;v_0,u) + \sum_{\scriptstyle u\in G}\sum_{\scriptstyle 
\omega:v_0\to u\atop \scriptstyle \omega\not\subseteq G_0, 
|\omega|\leq n} p_G(\omega)\,.
\end{equation}
Using \eqref{eq.4.4} and that 
$$
\sum_{\scriptstyle u\in G_0} \sigma_u \leq 2\,|\overline{G}_0|\,,
$$
the first sum on the right-hand side in 
(\ref{sumrule}) can be estimated from above by 
\begin{equation}\label{firstterm}
2\,|\bar G_0|\,\widetilde Q_G(n;v_0,v_0)\,.
\end{equation}
The last sum, on the other hand, can be estimated as follows. Any walk $\omega$ starting at $v_0$ which is not contained in $G_0$ can be 
decomposed in a unique way into a (possibly trivial) walk $\omega':v_0\to  v_0$ which is contained in $G_0$, followed by a walk 
$\omega'':v_0\twoheadrightarrow v$ which is likewise contained in $G_0$ but does not return to $v_0$ and such that 
$v\in\partial_{in} G_0$, and finally a step from $v$ to a vertex $v'\in\partial_{out} G_0$ and a walk $\omega'''$ starting at $v'$. 
Obviously, the lengths of $\omega', \omega''$ and $\omega'''$ sum up to at most $n$ and 
$p_G(\omega) = p_G(\omega')p_G(\omega'')\sigma_v^{-1}p_G(\omega''')$. Hence, 
an upper bound on the last term in \eqref{eq.4.5} is obtained by relaxing the constraint $|\omega'|+|\omega''|+|\omega'''|=n$ to 
$|\omega'|,|\omega''|,|\omega'''|\leq n$, in which case the sum factorizes into three terms: summation over $\omega'$ contributes 
a factor $\sigma_{v_0}\widetilde Q_G(n;v_0,v_0)$, summation over  $\omega'''$ is bounded by $n+1$ by \eqref{sumrule}, whereafter summation over $\omega'', v$, and $v'$ gives a factor $\sigma_{v_0}^{-1}{\mathscr C}_{G,G_0}(n;v_0)$. Hence, we have
 \begin{equation}\label{eq.4.6}
\sum_{\scriptstyle u\in G}\sum_{\scriptstyle \omega:v_0\to u\atop \scriptstyle \omega\not\subseteq G_0, 
|\omega|\leq n} p_G(\omega)\leq  (n+1)\,\widetilde Q_G(n;v_0,v_0){\mathscr C}_{G,G_0}(n;v_0)\,.
\end{equation}  
Using eqs. \eqref{eq.4.5}, \eqref{firstterm}, and  \eqref{eq.4.6} we finally arrive at 
$$
n+1 \leq  \widetilde Q_G(n;v_0,v_0)\left\{2\, |\bar G_0| + (n+1)\,{\mathscr C}_{G,G_0}(n;v_0)\right\}\,,
$$
which implies \eqref{low2} since  $\sigma_{v_0}\widetilde Q_G(n;v_0,v_0) \leq Q_G(n;v_0)$.
\end{proof}

Note that the relation of ${\mathscr C}_{G,G_0}(n;v_0)$ to an exit probability shows that it is bounded by $\sigma_{v_0}$. Clearly, ${\mathscr C}_{G,G_0}(n;v_0)$ is an increasing function of $n$ so the limit 
\begin{equation}\label{conducdef1}
{\mathscr C}_{G,G_0}(v_0) = \lim_{n\to\infty} {\mathscr C}_{G,G_0}(n;v_0) = \sup_{n\geq 1} {\mathscr C}_{G,G_0}(n;v_0)\,,
\end{equation}
exists and is, by definition, the \emph{effective conductance} of $G$ between $v_0$ and the complement of $G_0$. The \emph{effective resistance} of $G$ between $v_0$ and the complement of $G_0$ is defined as 
\begin{equation}\label{resistdef1}
{\mathscr R}_{G,G_0}(v_0) = \left({\mathscr C}_{G,G_0}(v_0)\right)^{-1}\,.
\end{equation}
Clearly, \eqref{low2} then implies 
\begin{equation}\label{Qlowerbound}
Q_G(n;v_0,v_0) \geq \sigma_{v_0}\left(\frac{2|\bar G_0|}{n+1} + {\mathscr R}_{G,G_0}(v_0)^{-1}\right)^{-1}\,,
\end{equation}

Given graphs $G$ and $G_0$ as above let us define $\hat G_0$ to be the graph obtained from $\bar G_0$ by identifying all vertices in $\partial^{\rm out}G_0$ with a single new vertex $v_1$ and leaving out all edges with both endpoints in $\partial^{\rm out}G_0$.\footnote{It should be noted that $\hat G_0$ may contain multiple edges, but the reader may easily verify that all considerations in the present subsection apply with obvious modifications also to graphs with multiple edges.} It is then clear that ${\mathscr C}_{G,G_0}(v_0) = {\mathscr C}_{\hat G_0,\hat G_0-v_1}(v_0) $, which is called the conductance of $\hat G_0$ between $v_0$ and $v_1$. It will also be denoted by ${\mathscr C}_{\hat G_0}(v_0,v_1)$. Since here $\hat G_0$ can be any finite graph, this defines the conductance between any two different vertices $v_0, v_1$ in a finite connected graph $H$ by
\begin{equation}\label{conducdef2}
{\mathscr C}_{H}(v_0,v_1) \,= \, \sum_{\omega:v_0\leftrightarrow v_1}  \sigma_{v_0} p_H(\omega)\,,
\end{equation} 
where we use the notation $\omega:v_0\leftrightarrow v_1$ to denote a path from $v_0$ to $v_1$ that does not hit the end-vertices at intermediate steps. Clearly, ${\mathscr C}_{H}(v_0,v_1)$ is symmetric in $v_0$ and $v_1$. 

Recalling \eqref{lemma1} and noting that any path $\omega: v_0\twoheadrightarrow v_1$ can uniquely be decomposed into a path $\omega': v_0\to v_0$ (possibly trivial) not hitting $v_1$ and a path $\omega'': v_0\leftrightarrow v_1$ we get that 
\begin{equation}\nonumber
 {\mathscr C}_{H}(v_0,v_1) \sum_{\scriptstyle 
 \omega':v_0\to v_0\atop \scriptstyle v_1\notin\omega} p(\omega') \sigma_{v_0}^{-1} \,=\, 1\,, 
\end{equation}
which implies that the resistance ${\mathscr R}_{H}(v_0,v_1) := \big({\mathscr C}_{H}(v_0,v_1)\big)^{-1}$ can be expressed as 
\begin{equation}\label{resistdef2}
{\mathscr R}_{H}(v_0,v_1) =  \sum_{\scriptstyle \omega:v_0\to v_0
\atop \scriptstyle v_1\notin\omega} p(\omega)\sigma_{v_0}^{-1}\,.
\end{equation}

The relation of these definitions to the physical notion of conductance and resistance in electrical networks is perhaps not obvious 
at this stage. From \eqref{conducdef2} it is clear that conductance and resistance of a single edge are equal to $1$ and it is simple to verify, using \eqref{conducdef2} and \eqref{resistdef2}, that the standard laws for composing reststances in a series or in parallel hold. We refer to Ch. 2 of \cite{LyoPer} for a more general and detailed account of these aspects, including Rayleigh's Monotonicity Principle which states that the effective resistance is a non-decreasing function of the edge-resistances. In our case of unit edge resistances this principle implies that contracting an edge in $H$ that does not connect $v_0$ and $v_1$, i.e. deleting the edge and identifying its end-vertices, reduces the resistance  ${\mathscr R}_{H}(v_0,v_1)$ or leaves it unchanged. 

We remark  that the inequality \eqref{cor2} can now be rewritten as 
\begin{equation}\label{resistest2}
{\mathscr R}_{H}(v_0,v_1)  \,\leq \, d_H(v_0,v_1)
\end{equation}
for any pair of different vertices $v_0, v_1$ in a finite connected graph $H$. Moreover, we have that equality holds if $H$ is a tree, i.e.
\begin{equation}\label{treeresist}
{\mathscr R}_{T}(v_0,v_1)  \,=\, d_T(v_0,v_1)\quad \mbox{\rm if $T$ is a finite tree}.
\end{equation}  
Indeed, if $d_T(v_0,v_1)=1$ only one path consisting of the edge connecting $v_0$ and $v_1$ contributes on the right-hand side of \eqref{conducdef2} and gives $1$. If $d_T(v_0,v_1) = k\geq 2$, let $v$ be a vertex in the interior of the unique path connecting $v_0$ and $v_1$ and let $T_1$ be the sub-tree of $T$ spanned by $v$ and its descendants when considering $v_0$ as the root of $T$, and let $T_0$ be the tree spanned by the remaining vertices and $v$. Then $T_0$ and $T_1$ only share the single vertex  $v$ and hence by the law of coupling resistances in a series we have 
$$
{\mathscr R}_{T}(v_0,v_1)  \,=\, {\mathscr R}_{T_0}(v_0,v)  \,+\,{\mathscr R}_{T_1}(v,v_1)\,.
$$
The claim \eqref{treeresist} now follows trivially by induction.

 Given a rooted graph $G$, let $\mathscr R (R)$ denote the resistance between the root $v_0$ and the complement of the ball $B_R(G)$ of radius $R$ around the root, and assume that for some $\kappa\geq 0$ it holds for $R$ large that
$$
{\mathscr R}(r) \geq \mbox{const.}\, R^\kappa\,.
$$
Note that \eqref{resistest2} implies the constraint
\begin{equation}\label{resistexpest}
\kappa\,\leq \, 1\,.
\end{equation}
If the Hausdorff dimension $d_h$ of $G$ exists we obtain, by choosing $G_0 = B_R(G)$ in \eqref{low2} where $R= \lfloor n^{\frac{1}{d_h+\kappa}}\rfloor$, that 
$$
Q_G(n;v_0,v_0) \, \geq \, \mbox{const.}\, n^{\frac{
\scriptstyle \kappa}{\scriptstyle d_h+\kappa}}
$$
and hence,
\begin{equation}\label{dimineq2}
d_s\leq \frac{2d_h}{d_h+\kappa}\,,
\end{equation} 
thus supplementing the lower bound \eqref{dimineq}.

If $G_0^K, K=1,2,3,\dots,$ is a sequence of finite connected graphs as in Theorem~\ref{thm2} 
containing a fixed vertex $v_0$ and such that the graph distance from $v_0$ to $\partial_{in}G_0^K$ tends to infinity 
as $K\to\infty$, then $\sigma_{v_0}^{-1}\lim_{K\to\infty} {\mathscr C}_{G,G_0^K}(v_0)$ exists and equals the escape probability from $v_0$, that is the 
probability for an infinite walk starting at $v_0$ never to return to $v_0$.  Hence, by the discussion of recurrency in Section \ref{sec:5.1}, we conclude that this quantity 
vanishes exactly if $G$ is recurrent. In particular, we get that if $G$ is recurrent then  ${\mathscr C}_{G,G_0^K}(n;v_0)$ converges to $0$ uniformly in $n$ as $K\to\infty$. In order to exploit 
Theorem~\ref{thm2} we need more detailed information on the decay rate of ${\mathscr C}_{G,G_0^K}(n;v_0)$ or ${\mathscr C}_{G,G_0^K}(v_0)$ for an appropriate choice of $G_0^K$. This is a non-trivial problem for general graphs, but if $G$ is a tree we can make use of \eqref{treeresist} as will be seen. 

We are now in a position to apply the previous results to the case of generic random trees and prove the desired upper bound on their spectral 
dimension. Recalling the one-spine character of the generic trees we let $T$ be such a tree and aim at applying Theorem~\ref{thm2} with $G_0^K$ equal to the sub-tree spanned by the vertices of the (finite) branches  rooted at the spine vertices $s_1,\dots, s_K$.

Denoting as previously by $Br_i$ the union of the branches rooted at a fixed spine vertex $s_i$, we then have 
\begin{equation}\label{sizeG0K}
|\bar G_0^K| = \sum_{i=1}^K |Br_i| +K+1\,,
\end{equation}
and obviously, $|\bar G_0^K| \geq |B_K|$. We can now use \eqref{bar-sing} to estimate the growth rate of $|G_0^K|$ by first establishing the following lemma. 

\begin{lemma}\label{lemma5} There exist constants $\bar c>0$ and $u_0>0$ such that for all $K\geq 1$ and all $u>u_0$ the following inequality holds:
\begin{equation}\label{simpleest}
\nu\left(\{T:\,K^{-2}|\bar G_0^K| \geq u\}\right)\;\; \leq\;\; \frac{\bar c}{\sqrt u}\,.
\end{equation}
\end{lemma}   
\begin{proof}
By \eqref{sizeG0K} it clearly suffices to show \eqref{simpleest} with  $X_K= \sum_{i=1}^K |Br_i|$ replacing $|\bar G_0^K|$. 
We take as starting point the following inequality which can be found, e.g., in \cite[Section 8.7]{Breiman}:
\begin{equation}\label{ineqBr}
\nu\left(\{X_K\geq u\right\})\;\; \leq\;\; \alpha u\int_0^{\frac{1}{u}} \left(1-\mbox{Re}\,\bar Z(e^{iv})^K\right)dv\,,
\end{equation}
where $\alpha$ is a universal constant and $\bar Z(e^{iv})^K$ is the characteristic function of $X_K$ as a sum 
of $K$ independent and identically distributed (i.i.d.) 
random variables. From \eqref{bar-sing} we have
$$
\bar Z(e^{iv}) = e^{\bar c_0\sqrt{1-e^{iv}} + O(|1-e^{iv}|)}
$$
and hence, for $v>0$ sufficiently small,
\begin{eqnarray*}
\mbox{Re}\,\bar Z(e^{iv})^K &=& \mbox{Re}\, e^{K\bar c_0\sqrt{1-e^{iv}} + K O(|1-e^{iv}|)} \\
 &=& \mbox{Re}\, e^{K\bar c_0\sqrt{-iv+O(v^2)} + K O(v)} \\ 
  &=& \mbox{Re}\, e^{\frac{1}{\sqrt 2}\bar c_0 K(1-i)\sqrt{v}\sqrt{1+O(v)} + K O(v)} \\
&=&  e^{\frac{1}{\sqrt 2}\bar c_0 K\sqrt{v} + K O(v)}\cos\left(\frac{1}{\sqrt 2}\bar c_0 K\sqrt{v} + K O(v)\right)\,. 
\end{eqnarray*}
By Taylor expanding the exponential and cosine functions it follows that there exists a $\delta>0$ such that 
$$
1-\mbox{Re}\,\bar Z(e^{iv})^K = \frac{1}{\sqrt 2}\bar c_0 K\sqrt{v} + O(K^2v) \leq \bar c_0 K\sqrt{v}\,,\;\; 
\mbox{for}\;\; 0\leq K\sqrt v\leq \delta\,. 
$$
Using this estimate in \eqref{ineqBr} we get
$$
\nu(\{X_K\geq u\})\;\; \leq\;\; \frac{2\alpha \bar c_0 K}{3\sqrt u}
$$
if $\frac{K}{\sqrt u}\leq \delta$. Upon replacing $u$ by $uK^2$ the claimed inequality follows for $u\geq \delta^{-2}$ with $\bar c = \frac 23\alpha \bar c_0$.
\end{proof}

Setting $K = 2^M$ and $u=M^3$ in \eqref{simpleest} we get 
$$
\sum_{M=1}^\infty \nu\left(\{ T: 4^{-M} |{\bar G_0}^{2^M}| 
\geq M^3\}\right) < \infty\,.
$$
 Hence, by the Borel-Cantelli lemma, we conclude that with probability $1$ it holds that
\begin{equation}\label{boundG0M}
 |\bar G_0^{2^M}| \leq  M^3 4^M
\end{equation}
 for $M> M_0$ where $M_0$ is an integer depending on $T$. Furthermore, given an arbitrary $K\geq 1$, we can choose $M$ such that $2^{M-1}\leq K\leq 2^{M}$ and conclude that 
\begin{equation}\label{boundG0K}
|\bar G_0^K|\, \leq \,  |\bar G_0^{2^{M}}| \leq  M^3 4^M \,\leq \, 4K^2 \Big(\frac{\ln K}{\ln 2}+1\Big)^3\,\leq \, c_1 K^2(\ln K)^3
\end{equation}
for $K> K_0$ , where $c_1>0$ is a numerical constant independent of $T$ while $K_0$ may depend on $T$. 

For a fixed $T\in {\cal T}_\infty$ as above let $n\geq 1$ be given and let $K=\lfloor  n^{\frac 13}\rfloor$  and assume $n$ is large enough so that $K>K_0$. It then follows from \eqref{low2} and \eqref{boundG0K} that 
\begin{equation}\label{lowerQn}
 Q_T(n;s_0,s_0)  \geq Q_T(K^3;s_0,s_0) \geq \Big(\frac{c_1 K^2(\ln K)^3}{K^3+1} + K^{-1}\Big)^{-1} \geq C_1 (\ln n)^{-3}\, n^{\frac 13}\,,
\end{equation}
where $C_1>0$ is a constant and we have also used that ${\mathscr C}_{T,G_0^K}(s_0) = (K+1)^{-1}$ by \eqref{treeresist}. Using 
(\ref{eq:s-dim}), the definition of $d_s$,  
it follows that $d_s\leq \frac 43$ with probability $1$.

 For the sake of completeness we mention that by elaborating on the upper bound \eqref{Qupperbound} on $Q_T(n;s_0,s_0)$ in a similar way as above and using the lower bound on $|B_R(T)|$ in \eqref{BoundsT}, 
we obtain an upper bound on $Q_T(n;s_0,s_0)$ analogous to \eqref{lowerQn}. More precisely, one gets that 
$$
C_1 (\ln n)^{-3} n^{\frac 13}\;\leq\; Q_T(n;s_0,s_0) \; \leq\; C_2  (\ln n)^2 \,n^{\frac 13}  
$$
for a suitable constant $C_2>0$, which of course also implies $d_s=\frac 43$.

\subsection{Spectral dimension of causal triangulations}\label{sec:5.4}

In this section we show that the spectral dimension of the CDT ensemble defined in Section \ref{sec:bijection} equals $2$. Thus, even though the Hausdorff dimensions of the CDTs and of the corresponding generic tree are identical the spectral dimensions are not. This is due to the higher connectivity of the CT which leads to  different behaviour of the resistance between the root and the complement of the ball around the root for large radius.

To obtain the upper bound $d_s\leq 2$ we make use of an argument that is most easily understood in terms of resistance estimates. Considering an arbitrary infinite causal triangulation G, let $G_R$ be the graph obtained from the ball $B_R(G)$ by collapsing its boundary to a single vertex $u$ (i.e. by contracting all the boundary edges). Then the resistance of $G$ between the central vertex $S_0$ and the complement of the ball $B_R(G)$ is identical to the resistance of $G_R$ between $S_0$ and $u$. By contracting the edges in each $S_r, r=1,\dots, R$, we obtain, by the Rayleigh Monotonicity Principle, a network with lower resistance between $S_0$ and $u$. On the other hand, this network is a series of resistances $\mathscr R_0, \mathscr R_1,\dots, \mathscr R_R$, where $\mathscr R_k$ is the resistance of $\Delta(\Sigma_k)$ unit resistances connected in parallel. Hence,
\begin{equation}\label{NWest}
{\mathscr R}_{G,B_R}(S_0) \, \geq \, \sum_{k=0}^R \frac{1}{\Delta(\Sigma_k)}\,.
\end{equation}
In order to make use of \eqref{Qlowerbound}
we need a suitable lower bound on the resistance ${\mathscr R}_{G, B_R}(S_0)$. In view of \eqref{NWest} we 
therefore want to estimate the probability $\rho(\{G: \Delta(\Sigma_k) \,> \, K\})$. For this purpose we use \eqref{ballprob} to write 
\begin{equation}\label{rhoballmeasure1} 
\rho(\{G: B_R(G) =  B_R(G_0)\}) = D_{R}(G_0)\,2^{D_{R}(G_0)+1} 2^{-\|B_R(G_0)\|}\,
 \end{equation}
for any fixed infinite causal triangulation $G_0\in {\cal C}$, where it is used that $Z_0=\frac 12$ and $\zeta_0=\frac 14$ in \eqref{ballprob} for the uniform planar tree and also $\|B_R(G_0)\|= |B_{R+1}(\beta(G_0))|$.
 Since the number of causal triangulations of the annulus $\Sigma_{k}$ with $l_k$ vertices on the inner boundary, among which one is marked, and $l_{k+1}$ vertices on the outer boundary equals $\binom{l_k+l_{k+1} -1}{ l_k-1}$, it follows by a straight-forward combinatorial argument that 
\begin{align}
\rho(\{G: D_R = l_R\,,\, D_{R+1}=l_{R+1}\}) = l_{R+1}\,2^{- l_{R+1}}\binom{l_R+l_{R+1} -1}{ l_R-1}\frac{1}{R(R-1)}\Big(\frac{R-1}{2R}\Big)^{l_R}\,,
\end{align}
see \cite{Durhuus:2009sm} for details. Summing this identity over $l_R+l_{R+1} > K$ then yields 
\begin{equation}\label{rhoballmeasure2} 
\rho(\{G: \Delta(\Sigma_R) \,> \, K\}) = \frac{K+2R-1}{2R-1}\Big(1-\frac{1}{2R}\Big)^K\,.
\end{equation}
 Now, let $a>2$ be fixed and define 
$$
{\cal A}_R \,=\, \{G : \Delta(\Sigma_R) > aR\ln R\}\,.
$$
Then  \eqref{rhoballmeasure2} implies that  $\rho({\cal A}_R) \leq (1+ a\ln R)R^{-a}$, and hence 
$$
\sum_{R=1}^\infty \rho({\cal A}_R)\,<\, \infty\,.
$$
By the Borel-Cantelli lemma it follows that with probability $1$ it holds that 
$$
\Delta(\Sigma_R) \, \leq a R\ln R
$$
for $R$ large enough.  Consequently, for all such $G\in {\cal C}_\infty$ we have 
\begin{equation}\label{lowresist}
{\mathscr R}_{G, B_R}(S_0) \, \leq \, \mbox{\rm const.}\sum_{k=1}^R (ak\ln k)^{-1}\, \leq \, \mbox{\rm const.}\ln\ln R\,.
\end{equation}

Setting  $R = \lfloor n^{\frac 13}\rfloor$ in \eqref{Qlowerbound} it follows from the bound \eqref{boundG0K} and the previous estimate that

\begin{equation}\label{lowerQ}
Q_G(n;S_0,S_0) \,\geq \, \Big(\frac{c_1 R^2(\ln R)^3}{n+1} +(a\ln\ln R)^{-1}\Big)^{-1}  \,\geq \, C_1^\prime \ln\ln n
\end{equation}
for a suitable constant $C_1^\prime>0$ (depending on $G$). By the definition of $d_s$ this evidently implies the desired upper bound $d_s\leq 2$. 

To obtain a useful upper bound on $Q_G(n;S_0,S_0)$ the inequality \eqref{Qupperbound} is not applicable since it does not capture the dependence of $d_s$ on the behaviour of the resistance between the root and the complement of balls at large radii. In our particular case, however, the bound 
\begin{equation}\label{upperresist}
Q_G(n;S_0) \,\leq \, \sigma_{S_0}\, {\mathscr R}_{G,B_R}(S_0)\quad\mbox{\rm for $n\leq R$}\,,
\end{equation}
which follows immediately from the definitions of $Q_G$ and ${\mathscr R}_{G,B_R}(S_0)$, is sufficient, provided a suitable upper bound on the resistance can be obtained. In \cite{curien2019geometric} it is shown by an argument requiring rather advanced probabilistic techniques that the bound 
$$
{\mathscr R}_{G,B_R}(S_0)\,\leq \, e^{\mbox{\rm const.}\sqrt{\ln R}}
$$
holds for $R$ sufficiently large with probability $1$. Hence, it follows by setting $R=n^b$ in \eqref{upperresist}, where $b\geq 1$,  that with probability $1$ we have
$$
Q_G(n;S_0) \,\leq \, \sigma_{S_0}\, e^{\mbox{\rm const.}\sqrt{\ln n}}
$$
for $n$ large enough. Clearly, this implies the lower bound $d_s\geq 2$.

\section{Curvature and matter fields on the CDT}\label{sec:matterfields}

Modifications of the graph weights, by introducing terms that bias the vertex degree, or adding extra degrees of freedom (often referred to as `matter fields') to the graphs, might lead to different critical behaviour.  Some of these modifications can be analysed using the bijection $\beta$, but understanding of others remains seriously incomplete. Some examples are discussed here in the grand canonical ensemble framework. 

\subsection{Curvature}\label{subsec:curvature}

The simplest elaboration of the graph weights is to introduce an extra factor $q^{\sigma_{v}-6}$ for each vertex 
into the definition of $W_M$ \eqref{eqn:Wdefn}. This is the analogue of
including the Ricci scalar curvature term from the Einstein action in a continuum gravity theory. It is trivial at fixed topology in two dimensions because $\sum_{v\in G}(\sigma_v-6)$ is proportional to the Euler characteristic of $G$ so $W_M$ is simply multiplied by a constant.  

A different higher curvature term was proposed by  \cite{DiFrancesco:1999em}. Recalling the definition of  $\sigma_{bv}$ and  $\sigma_{fv}$ from Section \ref{subsec:CDT},
introduce the extra factor
\begin{align}
   Q(p,q)= \prod_{v\in G\backslash S_0} q^{\half\abs{\sigma_{fv}-2}}\,p^{\half\abs{\sigma_{bv}-2}}\label{eqn:curv_weight}
\end{align}
into the graph weight for the disk partition function \eqref{eqn:Wdefn}.
Decreasing $q,p$ from 1 enhances the weight for vertices which belong to 3 triangles each in the forward and backward direction and thus introduces a bias towards graphs that, at least locally, are regular triangulations of the plane.  (It was argued in \cite{Glaser:2016smx} that this is equivalent to the naive discretization of an extrinsic curvature squared term in the continuum gravity action.)   However, it is straightforward to see that the critical properties of $W_M$ \cite{Durhuus:2009sm} remain unchanged as follows.
\begin{figure}
\sidecaption[t]
a)  \includegraphics[scale=1.0, trim=0cm 0cm 3cm 0cm,clip]{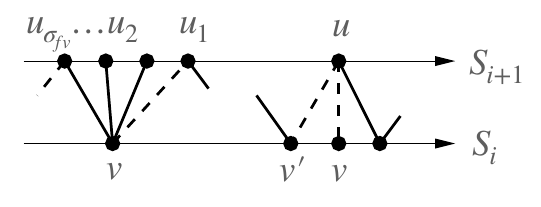}
b)\includegraphics[scale=1.0, trim=2.5cm 0cm 0cm 0cm,clip]{curvatureX.pdf}
\caption{The neighbourhood of a vertex $v\in G$ when a) $\sigma_{fv}>1$, and b) $\sigma_{fv}=1$. Dashed lines represent edges that are present in $G$ but not in the tree $T=\beta(G)$.}
\label{fig:ecurv}       
\end{figure}
Consider a vertex $v\in S_i(G)$ which shares edges with $u_m\in S_{i+1}(G), m=1,\ldots \sigma_{fv}$, see Fig. \ref{fig:ecurv} a. 
Then, making use of the bijection $\beta$, the $q$-dependent part of \eqref{eqn:curv_weight} is given by,
 \begin{align}\label{eqn:forward}
         \prod_{v\in G} q^{\half\abs{\sigma_{fv}-2}}= \prod_{\substack{v\in \beta(G):\, \sigma_{fv}\ge 2}} q^{\half(\sigma_{fv}-2)}\, \prod_{\substack{ u\in \beta(G):\, \sigma_{fu}=1}} q^{\half}.
    \end{align}
 The $p$-dependent part of \eqref{eqn:curv_weight} can be rewritten to give
\begin{align}\label{eqn:backward1}
     \prod_{v\in G\backslash S_0} p^{\half\abs{\sigma_{bv}-2}} =\prod_{\substack{v\in G\backslash S_0:\, \sigma_{bv}>1}} p^{\half(\sigma_{bv}-2)}\, \prod_{\substack{u\in G\backslash S_0: \, \sigma_{bu}=1}} p^{\half}.
\end{align}
Vertices  $u\in G$ with $\sigma_{bu}=1$ cannot be the leftmost descendant of a vertex $v\in\beta(G)$,  see Fig.\, \ref{fig:ecurv} a. So each vertex $v\in\beta(G)$ with $\sigma_{fv}\ge 2$ is associated with precisely $\sigma_{fv}- 2$ vertices $u\in G$ with $\sigma_{bu}=1$ which gives
\begin{align}\label{eqn:backward2}
    \prod_{\substack{u\in G\backslash S_0:\, \sigma_{bu}=1}}p^{\half}=\prod_{\substack{u\in \beta(G):\, \sigma_{fu}\ge 2}} p^{\half( \sigma_{fu}- 2)}.
\end{align}
Finally, vertices $u\in S_{i+1}(G)$ with $\sigma_{bu}>1$ must be the leftmost descendant of a vertex $v\in\beta(G)$,  see Fig.\, \ref{fig:ecurv} b; each vertex $v'\in S_i(G)$ with $\sigma_{fv'}=1$ then increments $\sigma_{bu}$ by one, so,
\begin{align}\label{eqn:backward3}
    \prod_{\substack{v\in G\backslash S_0:\, \sigma_{bv}>1}} p^{\half(\sigma_{bv}-2)}=\prod_{\substack{u\in \beta(G):\,\sigma_{fu}=1}} p^{\half}.
\end{align}
Combining \eqref{eqn:curv_weight}-\eqref{eqn:backward3} gives
\begin{align}
    Q(p,q)=\prod_{\substack{v\in \beta(G):\, \sigma_{fv}\ge 2}} (pq)^{\half(\sigma_{fv}- 2)}\, \prod_{\substack{u\in \beta(G):\, \sigma_{fu}=1}} (pq)^{\half},
\end{align}
so that each leaf in $T=\beta(G)$ gets a weight    $(pq)^{\half}$, and all other vertices a weight $(pq)^{\half(k_v-1)}$, where $k_v$ is the number of descendants.

Without loss of generality set $p=q$;
the recurrence \eqref{eqn:wrecurr} is then replaced by
\begin{align}\label{eqn:wqrecurr}
     \tilde w_h(g,q,z)&=q+\sum_{k=1}^\infty g^{2k}q^{k-1}(\tilde w_{h-1}(g,q,z))^k  \nonumber\\
    &=\frac{q+g^2(1-q^2)\tilde w_{h-1}(g,q,z)}{1-g^2q\tilde w_{h-1}(g,q,z)}.
\end{align}
It is easy to show (by direct substitution in \eqref{eqn:wqrecurr}, and using \eqref{eqn:wrecurr}) that the solution satisfying $\tilde w_1(g,q,z)= zg^{-1}$ is
\begin{align}
     \tilde w_h(g,z)= q-q^{-1}+\frac{q^{-1}}{1+g^2(1-q^2)}\,w_h\left(\frac{g}{1+g^2(1-q^2)}, qz+g(1-q^2)\right).
\end{align}
It follows from  \eqref{eqn:Adefn}  that the functions $\tilde w_h(g,z)$ for every $h$ are analytic  in the region 
\begin{align}
    A: \abs{g}<(1+q)^{-1},\;\abs{z}<1,
\end{align}
so this modification has no effect on the large graph behaviour and the model is in the same universality class for all $q$. 
It is clearly possible to define 
curvature dependent weights that take a form different from  \eqref{eqn:curv_weight}, for example   $q^{|\sigma_{fv}+\sigma_{bv}-4|}$ or $q^{(\sigma_{fv}+\sigma_{bv}-4)^2}$; but nothing is known about such systems.

\subsection{Dimers}\label{subsec:dimers}

Dimer models on fixed lattices exhibit a rich structure. In particular they have a critical point, at which the dimers condense, whose scaling limit is related to the Lee-Yang singularity and is described by a conformal field theory (CFT) \cite{Cardy:1985yy}.  
The model of dimers coupled to CTs, defined below, can be solved by a bijection, which is a generalization of $\beta$, to labelled trees \cite{Atkin:2012yt,Ambjorn:2012zx,Ambjorn:2014voa,Wheater:2021vnb}. There are new phases in which the interaction between the dimers and the graphs becomes strong and changes the universality class from the pure CT case of Section \ref{subsec:partitionfn}.
The steps to demonstrate this are outlined here, full details are in \cite{Wheater:2021vnb}.

\begin{figure}
\sidecaption[t]
 \includegraphics[scale=1.0]{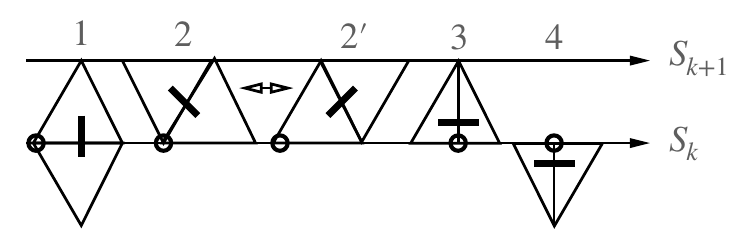}
\caption{The possible types of dimer on a CT, together with their index; the vertex that owns the dimer in each case is marked with a circle. The double arrow indicates the equivalence of types $2$ and $2'$ in $W_M$.}
\label{fig:dimertypes}       
\end{figure}
Dimers are objects that are dual to edges, and may be assigned freely subject to the \emph{dimer rule} that each triangle is allowed to contain only one dimer. The possible types of dimer on a CT, and the vertices which`own' them,  are shown in Fig. \ref{fig:dimertypes}; we will assume that there are no dimers allowed on the edges that are attached to the central vertex $S_0$, or that enter the boundary triangles of $G\in {\C C}^{(h)}$.  We denote  the dimer configurations  allowed by the dimer rule for $G\in\mathcal C$  by $\mathcal{L}(G)$; and the number of times a dimer of type $i$ appears in a dimer configuration $l\in\mathcal{L}(G)$ by $l_i$. The disk partition function for this model is then defined by assigning each dimer of type $i$ a weight $\xi_i$ and is given by
\begin{align}
    \label{eqn:Wdimerdefn}
    W_M(g,\{\xi\},z;h)&=\sum_{G\in {\C C}^{(h)}}\sum_{l\in\mathcal{L}(G)} \abs{S_{h-1}(G)}\, (z/g)^{\abs{S_{h-1}(G)}}\,  g^{\Delta(G)}\xi_1^{l_1}\xi_2^{l_2}\xi_{2'}^{l_{2'}}\xi_3^{l_3}\xi_4^{l_4}.
\end{align}
A crucial simplification \cite{Atkin:2012yt} is provided by the observation that, for every graph-and-dimer configuration with a type $2'$ dimer there is another which differs only by having a type $2$ dimer on a flipped edge, see Fig \ref{fig:dimertypes}. It follows that
\begin{align}
    W_M(g,\xi_1,\xi_2,\xi_{2'},\xi_3,\xi_4,z;h)=W_M(g,\xi_1,\xi_2+\xi_{2'},0,\xi_3,\xi_4,z;h),
\end{align}
so $\xi_{2'}$ can be set to zero, and the sum over dimer configurations limited to those with no type $2'$ dimers. It is convenient also to let type  0  mean no dimer and then define $\xi=\{\xi_0=1,\xi_1,\xi_2,\xi_3,\xi_4\}$. 
We see, by applying  the dimer rule, that each vertex $v$ can own at most one of the type $0,1,2,3$ dimers, and that if it owns a type $0,2,3$ dimer it can also own a type $4$, but that the combination of type $1$ and type $4$ is forbidden. So we assign to $v$ a two-component label $\ell_v$, denoting the dimers it owns, which can take values
\begin{equation}  \ell_v=(p_v,q_v) \in   \{0,1,2,3\}\times \{0,4\}\setminus \{(1,4)\}.\end{equation}
The dimer rule then implies a set of constraints $\mathcal N$  on the allowed labels, $\ell_v$ and $\ell_{v'}$, of neighbouring vertices, $v$ and $v'$, in $G\in\mathcal C$. 

\begin{figure}
\sidecaption[t]
  \includegraphics[scale=1.0]{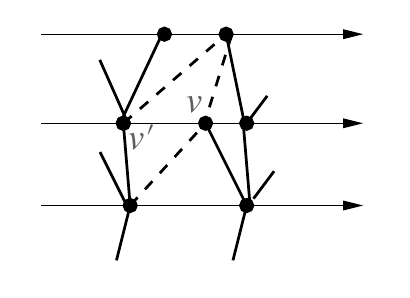}
\caption{Dashed lines represent edges that are present in $G$ but not in the tree $T=\beta(G)$. If $p_{v'}=3$ then $p_v=1$ is forbidden, which is a non-local constraint on $T$.}
\label{fig:clash}       
\end{figure}

Since each $v\in G$ is also a vertex of the corresponding tree  $T=\beta(G)$ the labels $\ell_v$ can also be associated with the vertices of $T$. Then $\mathcal N$ induces a set of rules for the labelling of $v\in T$ \cite{Wheater:2021vnb}. These rules are local in $T$
with one exception which is illustrated in Fig.\ref{fig:clash}. Two vertices which are nearest neighbours in $G$ can be arbitrarily widely separated in $T$; we deal with this non-local constraint in $T$ by forbidding  $p_v=3$ if $v$ is the most anticlockwise descendant of another vertex. 
The outcome is a model of dimers on CT that is equivalent to a model of labelled trees with a set of local constraints  $\widetilde {\mathcal N}$ on the labels; this model can in turn be solved by generalizing the methods of Section \ref{sec:gce}.

To solve this model by the decomposition used in Section \ref{sec:gce}, we have to keep track of the label $\ell$ assigned to  the vertex adjacent to the root of the tree $T$ (the root itself has no label). We then denote  the  allowed label configurations by $\mathcal L^\ell(T)$ and define the corresponding partition function
   \begin{align}
    w^\ell_h(g,\xi,z)&=
    \sum_{h'\le h}\,\sum_{T\in\C T^{(h')}} \sum_{l\in\mathcal L^\ell(T) }(z/g)^{\abs{S_{h}(T)}} \left(\prod_{v\in T \backslash r}
    g^{2(\sigma_v-1)} \xi_{p_v}\xi_{q_v}\right).\label{eqn:wldefn}
\end{align}
It is easy to show, by the arguments used in Section \ref{subsec:partitionfn}, that the disk partition function is given by
    \begin{align}
    \label{eqn:wdimerdefn}
    W_M(g,\xi,z;h)&=g^{2}
    \lderiv{z}w^{(0,0)}_h(g,\xi,z).
    \end{align}
    
    \begin{figure}
\begin{center}
  \includegraphics[scale=1.0]{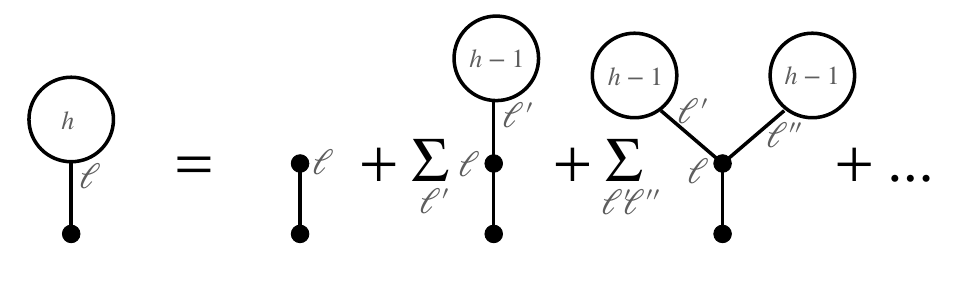}
 \end{center}
\caption{The decomposition of trees leading to \eqref{eqn:Frelns}. The sums over labels are constrained to satisfy the labelling rules $\widetilde {\mathcal N}$. }
\label{fig:labeltreecutting}       
\end{figure}

The trees contributing at height $h$ to \eqref{eqn:wldefn} can be decomposed into trees of height $h-1$ as shown in Fig. \ref{fig:labeltreecutting}. The local nature of the labelling rules ensures that it is only necessary to keep track of the labels at the first vertex of the component trees so the right hand side is still made up of geometric series, albeit more complicated than for the case without dimers \eqref{eqn:wrecurr}. The relationships obtained by resumming these series can be reduced to just two for $w_h^{(0,0)}$ and $w_h^{(2,0)}$ which take the form
\begin{align}
    w^{(0,0)}_{h+1}=  F^{(1)}\left( w^{(0,0)}_{h},w^{(2,0)}_{h}, g,\xi\right), \quad w^{(2,0)}_{h+1}= F^{(2)}\left( w^{(0,0)}_{h},w^{(2,0)}_{h}, g,\xi\right)\, ,\label{eqn:Frelns}
\end{align}
where
   \begin{align}
    F^{(1)}(f_1,f_2,g,\xi)&=\frac{1}
    {1-B}, \nonumber\\
    F^{(2)}(f_1,f_2,g,\xi)&=\frac{\xi_2}{1-B}\left(\frac{g^2f_1+g^2f_2+g^4\xi_3f_1}{1-g^4\xi_3\xi_4f_1}\right),\nonumber \\
    B=g^2\xi_1f_1+g^4\xi_4\xi_3f_1+&\frac{g^2(g^2\xi_3f_1+f_1+f_2)(1+g^2\xi_4f_1+g^2\xi_4f_2)}{1-g^4\xi_3\xi_4f_1}. \label{thend}
\end{align}
Iterating these relations with initial data $w^\ell_1(g,\xi,z)$ (which are easily computed) will give the partition functions for any finite height $h$. 

The grand canonical partition functions unconstrained by height are given by
 \begin{align}
   w^\ell_\infty(g,\xi)
    &=
    \sum_{T\in\C T} \sum_{l\in\mathcal L^\ell(T) } \left(\prod_{v\in T \backslash r}
    g^{2(\sigma_v-1)} \xi_{p_v}\xi_{q_v}\right).
    \label{eqn:winfdef}
\end{align}%
They are simply the fixed points of the recurrence equations so satisfy \eqref{eqn:Frelns} with $w^\ell_{h+1}=w^\ell_{h}=w^\ell_\infty$ for $\ell= (0,0),(2,0)$.
These equations reduce to a quartic equation for $w_\infty^{(0,0)}$ which can in principle be  solved exactly to determine the partition functions, although the explicit form of the solutions is not very useful and we will proceed in a different way. The solutions are also the limits, if these exist,
of the sequences $\mathcal S^\ell=(w^{\ell}_{h},h=1,2,\ldots)$.

As usual we expect that the functions $w^\ell_{h}(g,\xi,z)$ for every $h$, and the function $w^\infty_{h}(g,\xi)$, are analytic in some region 
\begin{align}
    A_\xi: \abs{g}<g_c(\xi),\;\abs{z}<z_c(\xi).
\end{align}
However this depends upon the sequences $S_\ell$ being smooth and converging to $w^\ell_\infty$.  If $\xi_i\ge 0, \forall i$ then $F^{(1)},F^{(2)}$ are positive convex functions of their first two arguments and the properties of the system are basically the same as for the pure CT model; the presence of the dimers does not drive any change in the geometrical properties, and there is no long-range cooperative behaviour of the dimers themselves. If the $\xi_i$ are sufficiently negative then convexity is no longer guaranteed so this behaviour can change. Indeed if they are too negative then 
the sequences $\mathcal S^\ell$ become oscillatory, and the model has no meaning in the statistical mechanical sense. New and interesting physics emerges for an intermediate set of dimer weights $\Xi$ which separate the CT-like region  from the oscillatory region.

We proceed by analysing the behaviour in the vicinity of $g=g_c(\xi)$ assuming the convergence properties discussed above. Adopting the simplified notation  $w^1_\infty\equiv w_\infty^{(0,0)}(g,\xi) $ and $w^2_\infty\equiv w_\infty^{(2,0)}(g,\xi)$ we have
\begin{align}
    w^k_\infty=F^{(k)}(w^1_\infty,w^2_\infty,g,\xi),\quad k=1,2,\label{eqn:finfreln}
\end{align}
and
\begin{align}
     \frac{\partial w^{k}_\infty}{\partial g}=\left((1- \mathbb T)^{-1} \right)^{kl}\frac{\partial \phantom{F}}{\partial g}F^{(l)}(w^1_\infty,w^2_\infty,g,\xi),
\end{align}
where
\begin{align}
    \mathbb T^{kl}=\frac{\partial \phantom{F}}{\partial w^{l}_\infty}F^{(k)}(w^1_\infty,w^2_\infty,g,\xi)     \equiv F_{l}^{(k)}.
\end{align}
At $g=0$, $\mathbb T$ vanishes; as $g\uparrow g_c(\xi)$,  $w^{k}_\infty$ develops non-analytic behaviour when the largest eigenvalue of $\mathbb T$ reaches one. We will denote: by $w^k_c$  the value of $w^{k}_\infty$ at $g_c(\xi)$; 
by  $(\lambda_1=1,\lambda_2)$ the eigenvalues at criticality of $\mathbb T=\mathbb T_c$; by $u_{1,2}$ the corresponding eigenvectors; and by $\overline u_{1,2}$ vectors orthogonal to $u_{1,2}$ respectively. $\mathbb T_c$ is not symmetric and it can be shown that, if the second eigenvalue $\lambda_2=1$, then for some values of $\xi$ it has Jordan normal form where $u_1$ is a regular eigenvector and $\mathbb T_c u_2=u_2+\epsilon(\xi) u_1$.

Expanding \eqref{eqn:finfreln}  about the critical point by setting
\begin{eqnarray}
    w^i_{\infty}&=&w^i_{c}+\phi^i,\nonumber\\
    g&=&g_c(\xi)-\Delta g,
\end{eqnarray}
we obtain
\begin{eqnarray} (1-{\mathbb{T}})^{ij}\phi^j&=& -\Delta g\left(\frac{\partial{F^{(i)}}}{\partial g}+\frac{\partial {\mathbb{T}}^{ij}}{\partial g}\phi^j\right)
+\frac{1}{2}F^{(i)}_{l k}\, \phi^l\phi^k\nonumber\\ &&+\frac{1}{3!}F^{(i)}_{kl m}\phi^k\phi^l\phi^m +O\left((\Delta g)^2,\phi^4\right).\label{eqn:master}
\end{eqnarray}
The different phases of the model arise when $\xi$ is tuned so that particular combinations of the coefficients in \eqref{eqn:master} vanish. With no constraints on these coefficients we obtain the generic case whose behaviour is the same as the pure CT model,
\begin{equation}
    \phi^i=-\phi_c(\xi) (\Delta g)^\frac{1}{2}  u^i_{1} +O(\Delta g)
    \,,\label{eqn:phinfg}
\end{equation}
where $\phi_c(\xi) >0$.
Imposing the constraint 
\begin{equation}
       \bar u^i_{2}(u^l_{1} F^{(i)}_{l k}u^k_{1}) = 0, 
    \label{tricritcond}
\end{equation}
leads to the \emph{Tri-critical} phase in which 
\begin{equation}
    \phi^i=-\phi_c(\xi) (\Delta g)^\frac{1}{3}  u^i_{1} +O((\Delta g)^\frac{2}{3})
    \,.\label{eqn:phinftc}
\end{equation}
Imposing the additional constraint
\begin{equation}
    \bar u^i_{2c}\frac{\partial{F^{(i)}}}{\partial g}= 0,
    \label{ddconstraint}
\end{equation}
leads to the \emph{Dense Dimer} phase. In this case $\phi_i$ behaves as in \eqref{eqn:phinfg} but other properties are different as we discuss next. The constraints \eqref{tricritcond} and \eqref{ddconstraint} are conditions on the dimer weights $\xi$; the first defines the region $\Xi$ introduced above.

The unconstrained disk partition functions \eqref{eqn:winfdef} can be written in the form
\begin{align}
     w^\ell_\infty(g,\xi)&=\sum_{N=1}^\infty w^\ell_N(\xi) g^{2N}.
\end{align}
It follows by general considerations from  \eqref{eqn:phinfg} and \eqref{eqn:phinftc} that
\begin{align}
    \lim_{N\to\infty}\frac{\log w^\ell_N(\xi)}{N} = -\log g_c(\xi);
\end{align}
therefore $\mu(\xi)=-\log g_c(\xi)$ is the thermodynamic free energy density, and  the dimer density is then defined by $  \chi(\xi)= -\xi_j\,\partial_{\xi_j} \log g_c(\xi)$.
In the pure CT phase  $\chi(\xi)$ is an analytic function, but as $\xi$ approaches $\xi_c\in\Xi$ 
it develops non-analytic behaviour,
\begin{align}
    \chi(\xi)= R_1(\xi)+R_2(\xi)\abs{\xi-\xi_c}^\sigma,\label{}
\end{align}
where $R_{1,2}$ are regular functions and $\sigma$ is usually called the dimer (density) exponent. If $\xi_c$ is in the Tri-critical phase, 
$\sigma=\half$ and $\chi$ itself remains finite but its derivative diverges; on the other hand if $\xi_c$ is in the  Dense Dimer phase then $\sigma=-\frac{1}{3}$, 
$\chi$ diverges at $\xi=\xi_c$ and the dimers condense, hence the name.

For each of the new phases, the local Hausdorff dimension $d_h$ and the scaling amplitude can be calculated. Although these phases only occur when dimer weights are negative, and individual graph weights can certainly be negative, it turns out that the theory at $g=g_c$ is nonetheless  described by a bijection to a labelled single spine tree \cite{Ambjorn:2014voa}. Thus, formally, at the level of expectation values it is possible to repeat the considerations of Section \ref{sec:hausdorff}; when $\lambda_2<1$ we find that $d_h=1$, which is not very interesting from the physical point of view. However there is still some freedom in the choice of $\xi$ which can be adjusted to impose the condition that $\mathbb T_c$ is not diagonalizable but has a Jordan normal form; this causes the finite trees attached to the spine in a typical graph  to become more bushy, and for both the Tri-critical and Dense Dimer phases it can be shown that 
\begin{align}
    \expect{B_R}{}=\mathrm{const.}\, R^3 +O(R^2),\label{eqn:dimerdh}
\end{align}
so the local Hausdorff dimension $d_h=3$. Unlike the pure CT case of Section \ref{subsec:partitionfn}, the disk partition functions at finite $h$ cannot be computed in closed form but their structure is very similar. 
Formally the scaling amplitudes are defined
 by setting $g=g_c(\xi)- \theta^{d_H}$, $y=y_c(\xi)(1-Y\theta\Lambda^{-d_H^{-1}}) $, $h=H\theta^{-1}\Lambda^{-d_H^{-1}}$ and taking $\theta\to0$, 
\begin{align}
 W_M^s(\Lambda,\xi,Y;H)= \lim_{\theta \to 0}  W_M(g_c(\xi)- \theta^{d_H},\,y_c(\xi)(1-\theta Y\Lambda^{-d_H^{-1}}),H\theta^{-1}\Lambda^{d_H^{-1}}).
\end{align}
In practice they can be computed in the scaling limit where the finite difference equations \eqref{eqn:Frelns} become solvable differential equations. 
  In the Tri-critical phase, $d_H=3$ and
   \begin{align}
      W_M^s(\Lambda,Y,H)&= \mathrm{const.}\, \frac{\sum_i\alpha_i e^{-\alpha_i\Lambda^\third H}}{\left(\sum_i (\alpha_i^2\Lambda^\third +\alpha_i Y)e^{\alpha_i\Lambda^\third H}\right)^2},\label{eqn:WsTC}
  \end{align}
  where the sum runs over the three cube roots of unity, $\alpha_i,\,i=1,2,3$.
  While this amplitude has an exponential decay envelope similar to the CT case \eqref{eqn:Ws}, it has oscillations superimposed. This reflects the negative dimer weights; taking the weights more negative destroys the convergence of the sequences $\mathcal S^\ell$ so the scaling limit no longer exists.
  In the Dense Dimer phase, $d_H=2$ and 
  \begin{align}
   W_M^s(\Lambda,Y,H)&= \mathrm{const.}\, \frac{\Lambda}{\left(\Lambda^\half\sinh H\Lambda^\half +Y\cosh H\Lambda^\half\right)^2}.\label{eqn:WsDD}
\end{align}
From the physical point of view this is the most interesting phase. At large $H$ the behaviour of $W^s_M$ deviates from the pure CT case by exponentially damped terms, as if there is some kind of matter interacting weakly with gravity. This is consistent with what is known numerically about other matter degrees of freedom (albeit with only positive weights) interacting with CT, which we discuss next.

\subsection{Ising spins}\label{subsec:ising}

The Ising model on a fixed two-dimensional square lattice was first solved by Onsager in 1944. It is very well known to exhibit a second order phase transition between a disordered phase at weak coupling (high temperature) and an ordered phase at strong coupling (low temperature);  the scaling limit at the critical coupling is a conformal field theory containing a single Majorana fermion. In contrast, no method to solve the model of Ising spins coupled to CTs is known. Numerical simulations \cite{Ambjorn:1999gi} indicate that the effect of the spins on the geometry is weak, and vice versa, so that the critical exponents do not change from the Onsager values and $d_H=2$; this is corroborated by weak coupling  expansions \cite{Benedetti:2006zy}, and is also true of the generalization  with the 3-state Potts model coupled to CT \cite{Ambjorn:2008jg}. There are only a few rigorous results which we now discuss briefly.

For Ising spins on a fixed CT, $G\in\mathcal C$, drawn from the ensemble with the measure $\rho$ \eqref{equivinfty}, the existence of a non-magnetised single phase at weak coupling, and multiple phases at strong coupling, was proved in \cite{Krikun:2008}. It was also proved that the critical coupling is almost surely independent of $G$ for such quenched systems. These estimates make extensive use of the tree correspondence. However, they do not easily extend to the CDT case, because then $G$ has to be chosen with a measure $\rho'\ne\rho$ that includes the effect of the spin partition function in the graph weight. This annealed model was studied in \cite{Napolitano:2015poa} where an upper bound on the radius of convergence, expressed as a function of the spin coupling strength, 
of the grand canonical ensemble was found.  It was also shown that at weak coupling the magnetization of the spin at the vertex $S_0$ on the disk vanishes. Similar results for the annealed $q$-state Potts model coupled to CTs were obtained in \cite{Hernandez:2016qmc}. There is still no proof known to us of the existence of multiple phases at strong coupling in the annealed models.

\section{Where next?\label{sec:wherenext}}

As we have seen through this article, there are several outstanding questions concerning CDT, and its extensions with extra degrees of freedom coupled to the geometry, which we draw together here. Firstly, the relationship between the grand canonical ensemble, where only expectation values can be computed, and the infinite graph ensemble, where the graphs dominating the measure have many properties in common, is unclear. The `observable' quantities naturally defined in one differ subtly from those naturally defined in the other. 
This relationship could be established by the construction of a limiting distribution of continuum surfaces corresponding to the grand canonical ensemble scaling limit. Secondly, much remains to be established for the Ising+CT annealed model. A proof that it magnetises at finite coupling, as seems almost certain from numerical work, would be significant progress; in the infinite graph ensemble this would involve establishing how the measure differs from $\rho$. Assuming the model does have a continuous phase transition, it would be interesting to establish rigorously whether the critical exponents are shifted from the Onsager values.

Obviously it is of interest to study the CDT model in higher dimensions.
In 3 dimensions some progress has been made, see, e.g., 
\cite{Ambjorn:2001br}, and it has been
proved that the number of different 3-dimensional CT grows exponentially
with the volume so the partition functions one would like to analyse
do converge for a range of coupling constants \cite{Durhuus:2014dbl}.
In 4 dimensions there are essentially only numerical results as described in 
other contributions 
to this book.  An important 
outstanding problem is to prove an exponential bound on the number of 4-dimensional CT as a function of volume.  In 
\cite{Durhuus:2014dbl} it is shown that such a bound would follow 
from an exponential bound on the number of unrestricted 3-dimensional 
triangulations of the sphere, as a function of volume.

\begin{acknowledgement}
JFW's research was funded by Research England. BD acknowledges support from Villum Fonden via the QMATH Centre of Excellence (Grant no. 10059). For the purpose of Open Access, the authors have applied a CC  BY public copyright licence to any Author Accepted Manuscript version arising from this article.

\end{acknowledgement}

\providecommand{\href}[2]{#2}\begingroup\raggedright\endgroup


\end{document}